\documentclass[11pt,draftclsnofoot,onecolumn]{IEEEtran}%
\usepackage{amssymb}    

\usepackage{color}      

\usepackage{cite}

\ifCLASSINFOpdf
   \usepackage[pdftex]{graphicx}
\else
   \usepackage[dvipdf]{graphicx}
\fi
\usepackage[cmex10]{amsmath}
\usepackage{url}

\newtheorem{definition}{Definition}
\newtheorem{example}{Example}
\newtheorem{lemma}{Lemma}
\newtheorem{remark}{Remark}
\newtheorem{condition}{Condition}
\newtheorem{corollary}{Corollary}
\newtheorem{proposition}{Proposition}

\newcommand{\lt}{\left}
\newcommand{\rt}{\right}
\newcommand{\cd}{\mathcal{C}}
\newcommand{\D}{\mathcal{D}}
\newcommand{\E}{\mathcal{E}}
\newcommand{\G}{\mathcal{G}}
\newcommand{\In}{\mathcal{I}}
\newcommand{\K}{\mathcal{K}}
\newcommand{\n}{\mathcal{N}}
\newcommand{\Sr}{\mathcal{S}}
\newcommand{\V}{\mathcal{V}}
\newcommand{\X}{\mathcal{X}}
\newcommand{\Y}{\mathcal{Y}}
\newcommand{\z}{\mathbb{Z}}
\newcommand{\f}{\mathbb{F}}
\newcommand{\A}{\overline{A}}
\newcommand{\As}[1]{\overline{A_{#1}}}
\newcommand{\B}{\overline{B}}
\newcommand{\C}{\overline{C}}
\newcommand{\I}{\overline{I}}
\newcommand{\Ba}{\overline{B_1}}
\newcommand{\Bc}{\overline{B_3}}
\newcommand{\Bd}{\overline{B_4}}
\newcommand{\U}{\overline{U}}
\newcommand{\fp}{\frac{p-1}{2}}
\newcommand{\fq}{\frac{q-1}{2}}
\newcommand{\mfp}{\frac{p+1}{2}}
\newcommand{\mfq}{\frac{q+1}{2}}
\newcommand{\ordm}[1]{|\text{--}#1|}

\begin{document}
%
\title{On the Ingleton-Violations in Finite Groups}
%
%
%

\author{Wei~Mao, Matthew~Thill, and~Babak~Hassibi,~\IEEEmembership{Member,~IEEE}
\thanks{Portions of this work were presented at the Forty-Seventh Annual Allerton Conference on Communication, Control, and Computing, 2009\cite{vioIngleton} and the 2010 IEEE International Symposium on Information Theory\cite{vioIngletonISIT}. The authors are with the Department of Electrical Engineering,
California Institute of Technology, Pasadena, CA 91125 USA (email: wmao@caltech.edu, mthill@caltech.edu, hassibi@caltech.edu). This work was supported in part by the National Science Foundation under grants CCF-0729203, CNS-0932428 and CCF-1018927, by the Office of Naval Research under 
the MURI grant N00014-08-1-0747, and by Caltech's Lee Center for Advanced Networking.}}

\maketitle

\begin{abstract}
Given $n$ discrete random variables, its entropy vector is the $2^n-1$ dimensional vector obtained from the joint entropies of all non-empty subsets of the random variables. It is well known that there is a one-to-one correspondence between such an 
entropy vector and a certain
group-characterizable vector obtained from a finite group and $n$ of
its subgroups \cite{Entropy_Group}. This correspondence may be useful for characterizing the space of entropic vectors and for designing network codes. If one restricts
attention to abelian groups then not all entropy vectors can be
obtained. This is an explanation for the fact shown by Dougherty et al
\cite{DFZ-insuff} that linear network codes cannot achieve capacity in
general network coding problems (since linear network codes form an
abelian group). All abelian group-characterizable vectors, and by fiat
all entropy vectors generated by linear network codes, satisfy a
linear inequality called the Ingleton inequality. It is therefore of interest to identify groups that violate the Ingleton inequality. In this paper, we
study the problem of finding nonabelian finite groups that yield
characterizable vectors which violate the Ingleton inequality. Using
a refined computer search, we find the symmetric group $S_5$ to be the
smallest group that violates the Ingleton inequality. Careful study of
the structure of this group, and its subgroups, reveals that it belongs to
the Ingleton-violating family $PGL(2,q)$ with a prime power $q \geq 5$, i.e.,
the projective group of $2\times 2$ nonsingular matrices with entries
in $\f_q$. We further interpret this family of groups, and their subgroups, using the theory of group actions and identify the subgroups as certain stabilizers. We also extend the construction to more general groups such as $PGL(n,q)$ and $GL(n,q)$. The families of groups identified here are therefore good candidates for
constructing network codes more powerful than linear network codes, and we discuss some considerations for constructing such group network codes.
\end{abstract}


%
\IEEEpeerreviewmaketitle

\section{Introduction}
\label{section:intro}
%
%
%
%

\IEEEPARstart{L}{et} $\n = \{1,2,\dots,n\}$, and let $X_1,X_2,\dots,X_n$ be $n$ jointly
distributed discrete random variables. For any nonempty set $\alpha
\subseteq \n$, let $X_\alpha$ denote the collection of random
variables $\{X_i:i\in\alpha\}$, with joint entropy $h_\alpha\triangleq
H(X_\alpha) = H(X_i;\,i\in\alpha)$. We call the ordered real $(2^n-1)$-tuple
$(h_\alpha:\emptyset \neq \alpha \subseteq \n) \in \mathbb{R}^{2^n-1}$ an
\emph{entropy vector}. The set of all entropy vectors
derived from $n$ jointly distributed discrete random variables is
denoted by $\Gamma_n^*$. It is not too difficult to show that the
closure of this set, i.e., $\overline{\Gamma_n^*}$, is a {\em convex cone}\cite{ZY97}.

The set $\overline{\Gamma_n^*}$ figures prominently in information theory
since it describes the possible values that the joint entropies of a
collection of $n$ discrete random variables can obtain. From a
practical point of view, it is of importance since it can be shown
that the capacity region of any arbitrary multi-source multi-sink {\em
  wired} network, whose graph is acyclic and whose links are discrete
memoryless channels, can be obtained by optimizing a linear function
of the entropy vector over the convex cone $\overline{\Gamma_n^*}$ and a set
of linear constraints (defined by the network)
\cite{yanyeungzhen07,hash07}. Despite this importance, the entropy
region $\overline{\Gamma_n^*}$ is only known for $n=2,3$ random variables
and remains unknown for $n\geq 4$ random variables. Nonetheless, there
are important connections known between $\overline{\Gamma_n^*}$ and
matroid theory (since entropy is a submodular\footnote{A set function $f$ defined on the subsets of $\n$ is \emph{submodular} iff $f_\alpha + f_\beta - f_{\alpha\cap\beta} - f_{\alpha\cup\beta}\geq0$ for all $\alpha,\beta \subseteq \n$.} function.)\cite{dofrze07}, determinantal inequalities
(through the connection with Gaussian random variables) \cite{shha08},
and quasi-uniform arrays \cite{chan01}. However, perhaps most
intriguing is the connection to finite groups which we briefly
elaborate below.

\subsection{Groups and Entropy}
\label{subsec:group}

Throughout this paper we use group theory notation defined in Section~\ref{section:notation}.
Let $G$ be a finite group, and let $G_1,G_2,\dots,G_n$ be $n$ of its
subgroups. For any nonempty set $\alpha \subseteq \n$, the group $G_\alpha
\triangleq \bigcap_{i\in\alpha}G_i$ is a subgroup of $G$. Define $g_\alpha = \log\frac{|G|}{|G_\alpha|}$. We call the ordered real
$(2^n-1)$-tuple $(g_\alpha: \emptyset\neq \alpha \subseteq \n) \in
\mathbb{R}^{2^n-1}$ a (finite) \emph{group characterizable vector}. Let
$\Upsilon_n$ be the set of all group characterizable vectors derived
from $n$ subgroups of a finite group.

The major result shown by Chan and Yeung in \cite{Entropy_Group} is
that $\overline{\Gamma_n^*} =
\overline{\mathrm{cone}(\Upsilon_n)}$, i.e., the closure of
$\Gamma_n^*$ is the same as the closure of the cone generated by
$\Upsilon_n$. Specifically, every group characterizable vector is
an entropy vector, whereas every entropy vector is arbitrarily
close to a scaled version of some group characterizable vector.

To show the first part of this statement, let $\Lambda$
be a random variable uniformly distributed on the elements of $G$ and define $X_i = \Lambda G_i$ (the left coset of $G_i$ in $G$ with representative $\Lambda$) for $i=1,\ldots,n$. Then $X_i$ is uniformly distributed on $G/G_{i}$ and $H(X_{i}) = \log\frac{|G|}{|G_i|}$. To calculate the joint entropy $h_\alpha = H(X_\alpha)$ for a nonempty subset $\alpha \subseteq \n$, let $\X_{\alpha}$ denotes the set of all coset tuples $\{(xG_i:i\in\alpha)\ |\ x\in G\}$. Consider the intersection mapping $\Theta_{\alpha}: \X_{\alpha} \to G/G_\alpha$, where for all $x\in G$,
\begin{equation}\label{equation:intersection_mapping}
\Theta_{\alpha}:(xG_i:i\in\alpha) \mapsto \bigcap_{i\in\alpha}xG_i = xG_\alpha.
\end{equation}
$\Theta_{\alpha}$ is a well defined onto function on $\X_{\alpha}$, and it is one-to-one since if $(xG_i:i\in\alpha)$ and $(x'G_i:i\in\alpha)$ are mapped to the same coset $xG_\alpha= x'G_\alpha$, then $x^{-1}x'\in G_\alpha$ and so $x^{-1}x'\in G_i$ for all $i$, which implies $(xG_i:i\in\alpha) = (x'G_i:i\in\alpha)$. So $H(X_\alpha) = H(\Theta_{\alpha}(X_\alpha))$, and as $\Theta_{\alpha}(X_\alpha) = \Lambda G_{\alpha}$, we have
\[ h_\alpha = H(\Lambda G_\alpha)) = \log\frac{|G|}{|G_\alpha|} = g_\alpha. \]
Thus indeed every group-characterizable vector is an entropy
vector. Showing the other direction, i.e., that every entropy vector
can be approximated by a scaled group-characterizable
vector is more tricky (the interested reader may consult
\cite{Entropy_Group} for the details). Here we shall briefly describe
the intuition.

Consider a random variable $X_1$ with alphabet size $N$ and probability
mass function $\{p_i,i=1,\ldots,N\}$. Now if we make $T$ copies of
this random variable to make sequences of length $T$, the entropy of
$X_1$ is roughly equal to the logarithm of the number of strongly typical
sequences, divided by $T$. These are
sequences where $X_1$ takes its first value roughly $Tp_1$ times, its
second value roughly $Tp_2$ times and so on. Therefore assuming that $T$ is
large enough so that the $Tp_i$ are close to integers (otherwise, we
have to round things) we may roughly write
\[ H(X_1) \approx \frac{1}{T}\log\left(\begin{array}{ccccc} & & T & &
  \\ Tp_1 & Tp_2 & \ldots & Tp_{N-1} & Tp_N \end{array} \right), \]
where the argument inside the $\log$ is the usual multinomial
coefficient. Written in terms of factorials this is
\begin{equation}
H(X_1) \approx \frac{1}{T}\log\frac{T!}{(Tp_1)!(Tp_2)!\ldots
  (Tp_N)!}.
\label{firstent}
\end{equation}
If we consider the group $G$ to be the symmetric group $S_T$, i.e.,
the group of permutations among $T$ objects, then clearly $|G| =
T!$. Now partition the $T$ objects into $N$ sets each with $Tp_1$ to
$Tp_N$ elements, respectively, and define the group $G_1$ to be the
subgroup of $S_T$ that permutes these objects {\em while respecting the
  partition}. Clearly, $|G_1| = (Tp_1)!(Tp_2)!\ldots (Tp_N)!$, which
is the denominator in (\ref{firstent}). Thus, $H(X_1) \approx
\frac{1}{T}\log\frac{|G|}{|G_1|}$, so that the entropy $h_{\{1\}}$ is a scaled
version of the group-characterizable $g_{\{1\}}$. This argument can be made
more precise and can be extended to $n$ random variables---see
\cite{Entropy_Group} for the details. We note, in passing, that this
construction often needs $T$ to be very large, so that the group $G$
and the subgroups $G_i$ are huge.

\subsection{The Ingleton Inequality}

As mentioned earlier, entropy satisfies submodularity and is connected to the notion of matroids.
A matroid is defined by a ground set $S$ and a rank function $r$ (written as $r(\{\cdot\})$ or $r_{\{\cdot\}}$) defined over subsets of $S$, that satisfies the following axioms:
\begin{enumerate}
\item $r$ is always a non-negative integer, and $r(U)\leq |U|$, $\forall U\subseteq S$.
\item $r$ is monotonic: if $U\subseteq W\subseteq S$, then $r(U)\leq r(W)$.
\item $r$ is submodular.
\end{enumerate}
Axioms 2) and 3), together with positiveness, are called the \emph{Shannon inequalities} for a set function. A matroid is defined in a way to extend the notion of a collection of vectors (in some vector space) along with the usual
definition of the rank. It is called {\em representable} if its ground set can be represented as a
collection of vectors (defined over some finite field) along with the usual rank function. Determining whether a matroid is representable or not is, in general, an open problem.

In 1971 Ingleton showed that for $n = 4$, the rank function $r$ of any representable matriod must satisfy the
inequality \cite{Ingleton}
\[ r_{12}+r_{13}+r_{14}+r_{23}+r_{24} \geq
r_1+r_2+r_{34}+r_{123}+r_{124} \]
(where for simplicity we write $r_{ij}$ and $r_{ijk}$ for
$r_{\{i,j\}}$ and $r_{\{i,j,k\}}$, respectively). In fact, these \emph{Ingleton inequalities}, together with the Shannon inequalities and their combinations, are the only inequalities the rank function of a representable matroid needs to satisfy (which are called linear rank inequalities) when $n=4$ (see \cite{Hammer-Ineq}). Furthermore, \cite{Hammer-Ineq} shows that the rank function of any representable matroid is necessarily an entropy vector, but not every linear rank inequality is respected by a general entropy vector. For example, there are entropy vectors that violate the Ingleton inequality (e.g. \cite{Hammer-Ineq,Matus-exVioIngleton}), so that entropy is generally not a representable matroid. Using non-representable matroids, \cite{DFZ-insuff} constructs
network coding problems that cannot be solved by linear network codes
(since linear network codes are, by definition, representable).

When $n\geq5$, there are many more linear rank inequalities besides the Shannon ones. But since the focus of this paper is the simplest case $n=4$ with only one such inequality, we refer the interested readers to the works of Kinser\cite{Kinser-NewIneq}, Dougherty \textit{et al}.\cite{DFZ-NewIneq5up,Dougherty-NewIneq6,DFZ-NewIneq8} and Chan \textit{et al}.\cite{Chan-TruncLinearIneq} for recent development in this area.

From this point on we shall only study the Ingleton inequality, with $n=4$. In the case of entropy vectors, it is written as
\begin{equation}\label{equation:Ingletonh}
h_{12}+h_{13}+h_{14}+h_{23}+h_{24} \geq h_1+h_2+h_{34}+h_{123}+h_{124}.
\end{equation}
The following sufficient condition is proposed in\cite{Hammer-Ineq} for four general random variables $X_1$, $X_2$, $X_3$ and $X_4$ to satisfy \eqref{equation:Ingletonh}:

\begin{lemma}\label{lemma:comm-info}
If there exists a random variable $Z$ that is a \emph{common information} for $X_{1}$ and $X_{2}$, i.e., $H(Z|X_{1}) = H(Z|X_{2}) = 0$ while $H(Z) = I(X_{1};X_{2})$, then \eqref{equation:Ingletonh} is satisfied.
\end{lemma}

In general common information does not exist for two arbitrary random variables, but when the entropies correspond to ranks of vector subspaces, their common information does exist\cite{Hammer-Ineq} and that is why representable matroids respect Ingleton. In Section~\ref{section:conditions} we will prove a similar condition for groups to satisfy Ingleton, by constructing a common information.

As $\overline{\Gamma_n^*} = \overline{\mathrm{cone}(\Upsilon_n)}$, we
know there must exist finite groups, and corresponding subgroups, such
that their induced group-characterizable vectors violate the Ingleton
inequality. In \cite{Chan-GroupCharEntropy} it was shown that abelian groups
cannot violate the Ingleton inequality, thereby giving an alternative
proof as to why linear network codes (and even the more general abelian group network codes (defined below)) cannot achieve capacity on arbitrary networks, as the underlying groups for linear network codes are abelian. So we need to focus on non-abelian groups and their connections to nonlinear codes. Note that in the context of finite groups, the Ingleton inequality can be rewritten as
\begin{equation}\label{equation:Ingletong}
|G_1||G_2||G_{34}||G_{123}||G_{124}| \geq |G_{12}||G_{13}||G_{14}||G_{23}||G_{24}|.
\end{equation}

\subsection{Group Network Codes}
\label{subsec:networkcodes}

A communication network is usually represented by a directed acyclic
graph $\G = (\V,\E)$, where the node set $\V$ and the edge set $\E$
model the communication nodes and channels respectively. Let
$\Sr \subset \V$ be the set of source nodes and
$\D(s)$ be the set of sink nodes demanding source $s$ for each $s \in \Sr$. For any
node $v$ and any edge $e$, $\In(v)$ and $\In(e)$ denote the sets of
incoming edges to $v$ and to the tail node of $e$, respectively.

A network code should include
\begin{enumerate}
\item the assignment of a symbol $Y_s$ from some alphabet $\Y_{s}$ for a source message at each source $s$;
\item the encoding of a symbol $Y_e$ in some alphabet $\Y_{e}$ at each edge $e$, from the symbols on $\In(e)$. Namely, $Y_e = \phi_e\left(Y_f:f\in\In(e)\right)$ for some deterministic encoding function $\phi_e$;
\item the decoding of the symbol $Y_s$ at each $u \in \D(s)$ for all sources $s$, i.e. $Y_s$ is uniquely determined from the symbols on $\In(u)$: $Y_s = \phi_{u,s}(Y_f:f\in\In(u))$ for some decoding function $\phi_{u,s}$.
\end{enumerate}
It is clear that at each edge $e$ the symbol $Y_{e}$ is a deterministic function of the source symbols $\{Y_s: s\in\Sr\}$, which is denoted by $\varphi_{e}$ and is called the \emph{global mapping} at $e$. Also the source random variables $\{Y_s: s\in\Sr\}$ are usually assumed to be independent and uniform on their respective alphabets.

For example, a linear network code is defined as follows: 1) for each $t\in\Sr\cup\E$, the alphabet $\Y_{t}$ is a vector space $F^{d_{t}}$ over a finite field $F$ with some finite dimension $d_{t}$; 2) all encoding/decoding functions are linear: if $t$ is an edge or a sink node, then the encoding/decoding function $\phi_t$ at $t$ can be written as
\[ \phi_t\left(Y_f:f\in\In(t)\right) = \sum_{f\in\In(t)}M_{t,f}Y_{f} \]
for some matrices $M_{t,f}\in F^{d_{t}}\times F^{d_{f}}$. Thus the global mappings at the edges are also linear.

Group network codes were first proposed by Chan in
\cite{Opt_GrpCodes,CapReg_LinAb}, considering the fact that finite
groups can generate the whole entropy region, and noting that linear
network codes are included as a special case. Suppose $G$ is a finite
group, $\{G_e:e\in\E\}$ and $\{G_s:s\in\Sr\}$ are some of its
subgroups. One can construct a network code with $\Y_{t} = G/G_{t}$ for each $t\in\Sr\cup\E$ if the following requirements are met:
\begin{enumerate}[\IEEEsetlabelwidth{\textbf(R3)}]
    \item[\textbf(R1)] \emph{Source independence}: $H(Y_{\Sr}) = \sum_{s\in\Sr}H(Y_{s})$, which means that the cardinalities of $G/G_{\Sr}$ and $\prod_{s\in\Sr}\Y_{s}$ (the Cartesian product of the source alphabets) are equal, where $G_{\Sr} \triangleq \bigcap_{s\in\Sr}G_s$. This is equivalent to $\prod_{s\in\Sr}|G_s| = |G|^{|\Sr|-1}|G_{\Sr}|$.
    \item[\textbf(R2)] \emph{Encoding}: $\forall e\in\E$, $\bigcap_{f\in\In(e)}G_f \leq G_e$.
    \item[\textbf(R3)] \emph{Decoding}: $\forall s\in\Sr$, $\bigcap_{f\in\In(u)}G_f \leq G_s$ for each $u \in \D(s)$.
\end{enumerate}
Moreover, the entropy vector for the network symbols $\{Y_{t}: t\in\Sr\cup\E\}$ is characterizable by the group $G$ and its subgroups $\{G_{t} : t\in\Sr\cup\E\}$.

In Section \ref{section:gnc} we discuss some important considerations necessary when constructing group network codes, such as how to ensure the source independence requirement (R1) above. Appendix~\ref{section:gnc-detail} provides further detailed discussions of group network codes, including the encoding/decoding construction, the induced entropy vectors, as well as the inclusion of linear network codes. We remark that any group network code constructed from an Ingleton-violating group induces entropy vectors that violate the Ingleton inequality, so potentially they are more powerful than linear network codes. As we shall see further below, this is certainly true of the group network codes that can be obtained from the Ingleton-violating families in this paper---the $PGL$ and $GL$ groups, especially since both contain linear network codes as subgroups.

\subsection{Discussion}

Since we know of distributions whose entropy vector violates the
Ingleton inequality, we can, in principle, construct finite groups
whose group-characterizable vectors violate Ingleton. Two such
distributions are Example 1 in \cite{Matus-exVioIngleton}, where the
underlying distribution is uniform over 7 points and the random
variables correspond to different partitions of these seven points,
and the example on page 1445 of \cite{ZY98}, constructed from finite
projective geometry and where the underlying distribution is uniform
over $12\times 13 = 156$ points. Unfortunately, constructing groups
and subgroups for these distributions using the recipe of Section~\ref{subsec:group} results in $T = 29\times 7 = 203$ and $T = 23\times
156 = 3588$, which results in groups of size $203!$ and $3588!$, which are too
huge to give us any insight whatsoever.

These discussions lead us to the following questions.

\begin{enumerate}
\item Could the connection between entropy and groups be a red
  herring? Are the interesting groups too large to give any insight
  into the problem (e.g., the conditions for the Ingleton inequality
  to be violated)?
\item What is the smallest group with subgroups that violates the
  Ingleton inequality? Does it have any special structure?
\item Can one construct network codes from such Ingleton-violating
  groups?
\end{enumerate}

In this paper we address the first two questions, and try to lay some groundwork for answering the third. We
identify the smallest group that violates the Ingleton inequality---it is the symmetric group $S_5$, with 120 elements. Through a thorough investigation of the structure of its subgroups we conclude that it belongs to the family of groups $PGL(2,q)$, with $q\geq5$ being a power of a prime. ($PGL(2,5)$ is isomorphic to $S_5$.) We therefore believe that the connection to groups is not a red herring and that there may be some benefit to it.

Having a ``recipe'' for Ingleton violations, we generalize the family in two directions. Since $PGL(2,q)$ is the quotient group of $GL(2,q)$ modulo the scalar matrices, we explore the subgroups in $GL(2,q)$ and discover several new families of Ingleton violations. On the other hand, the projective general linear group $PGL(n,q)$ can be viewed as the image of a permutation representation induced by the action of the general linear group $GL(n,q)$ on its projective geometry. It turns out that in this context, the Ingleton-violating subgroups of the family $PGL(2,q)$ all have nice interpretations: each of them is the stabilizer for a set of points in the projective geometry. Based on this viewpoint we obtain a few new families of Ingleton violations, including the groups $PGL(n,q)$, $GL(n,q)$, and further give an abstract construction in general 2-transitive groups.

As mentioned in Section~\ref{subsec:networkcodes} we can use these Ingleton-violating groups to contruct network codes, which have the potential of performing better than linear network codes. However, designing the subgroups for a desirable code is not a trivial task, for example we need to satisfy (R1)--(R3) of the previous subsection. We study the source independence requirement for the subgroups, and give some directions on how to construct them.

%

Before we proceed to present the details of our results, we would like to mention some recent developments after our first paper\cite{vioIngleton} on this subject. In \cite{Boston-Nan-NewVioIngelton}, Boston and Nan mainly study symmetric groups and discover many new Ingleton violations in the related groups. Furthermore, using the same group action theoretic approach as above (specifically, designing the subgroups to be the stabilizers of certain sets of points\footnote{In fact, in the original paper of Chan and Yeung\cite{Entropy_Group} the same type of subgroups are also used in to show that every entropy vector can be approximated by a scaled group-characterizable vector.}), they systematically construct subgroups of a symmetric group to violate Ingleton. Many of these new violations are quite effective (see Section~\ref{subsec:efc-Ingleton} for more discussions). Also, while all the Ingleton-violating groups in this paper are non-solvable, \cite{Boston-Nan-NewVioIngelton} shows that there do exist solvable groups that violate Ingleton. Paajanen\cite{Paajanen-NilGrpIngleton}, however, focuses on the subclasses $p$-groups and nilpotent groups and shows that with some technical conditions they satisfy Ingleton. Recall that we have the hierarchy of finite groups
\[ \textbf{Cyclic groups}\ \subset\ \textbf{Abelian groups}\ \subset\ \textbf{Nilpotent groups}\ \subset\ \textbf{Solvable groups}\ \subset\ \textbf{All groups} \]
and that every nilpotent group is a direct product of groups, each of which is a $p$-group for a distinct $p$. Now roughly speaking, we have a guideline for what class of groups one needs to explore to violate Ingleton. For linear rank inequalities in higher dimensions, \cite{Markin-Thomas-Oggier-vioIneq5} considers the case $n=5$ and obtains some results on the groups that satisfy/violate some of these inequalities.
\\

The rest of the paper is organized as follows. Section~\ref{section:notation} provides necessary notations. Section~\ref{section:conditions} describes the computer search process of Ingleton-violating groups and proves several conditions that help pruning the search. Having found the smallest violation instance, Section~\ref{section:presentation} studies its structure using group presentations. Section~\ref{section:pgl} then generalizes the instance to an Ingleton-violating family in $PGL(2,p)$, and then to $PGL(2,q)$, through explicitly constructing the subgroups in the format of matrices. Furthermore, the preimage group $GL(2,q)$ is also examined and 15 new families of Ingleton violating subgroups are identified, in Section~\ref{section:gl}. The original family has a deep relation to the theory of group actions, as disclosed in the more abstract Section~\ref{section:Asch}, which leads to several new violation constructions in this framework. Section~\ref{section:gnc}, however, considers using these groups to build group network codes and obtains some results in that regard. Section~\ref{section:conclusion} concludes this paper.

\section{Notation}
\label{section:notation}

We use the following abstract algebra notations. These are fairly standard (and follow Dummitt and Foote \cite{absAlg}). The interested reader, who may not be familiar with all the concepts below, may refer to \cite{absAlg}, or any other standard abstract algebra textbook.
\begin{description}[\IEEEsetlabelwidth{$\langle g_1,\dots,g_m \rangle$, $\langle S \rangle$}\IEEEusemathlabelsep]
\item[$|G|$] the order (cardinality) of the set/group $G$.
\item[$|g|$] the order of element $g$ = smallest positive integer $m$ s.t. $g^m = 1$.
\item[$x^g$] the conjugate of element $x$ by element $g$ in $G$: $x^g=g^{-1}xg$. (No confusion with the powers of $x$ as $g$ is an element of $G$.)
\item[$X^g$] the conjugate of subset $X$ by element $g$ in $G$: $X^g=\{x^g:x\in X\}$.
\item[$G \cong H$] the group $G$ is isomorphic to the group $H$.
\item[$H \leq G$] $H$ is a subgroup of $G$.
\item[$H \trianglelefteq G$] $H$ is a normal subgroup of $G$, i.e., $H^{g} = H$, $\forall g\in G$.
\item[$gH$] the left coset of the subgroup $H$ in $G$ with representative $g$.
\item[$G/H$] the set of all left cosets of subgroup $H$ in $G$. When $H \trianglelefteq G$, $G/H$ is a group, called the factor group or quotient group.
\item[$HK$ or $H\cdot K$] the ``set product'' of $H,K\subseteq G$: $HK = \{hk:h\in H, k\in K\}$.
\item[$H\times K$] the direct product of groups $H$ and $K$. The elements are the pairs $\{(h,k):h\in H, k\in K\}$ and $(h_{1},k_{1})(h_{2},k_{2}) = (h_{1}h_{2},k_{1}k_{2})$. 
\item[$G^n$] the direct product of $n$ copies of the group $G$.
\item[$H\rtimes K$] the semidirect product of groups $H$ and $K$. The elements are the same as $H\times K$, but $(h_{1},k_{1})\cdot(h_{2},k_{2}) = \lt(h_{1}\cdot\varphi(k_{1})(h_{2}),k_{1}k_{2}\rt)$ where $\varphi$ is a homomorphism of $K$ into the automorphism group of $H$.
\item[$\langle g_1,\dots,g_m \rangle$, $\langle S \rangle$] the group generated by the elements $g_1,\dots,g_m$, and by the set $S$.
\item[$G=\langle S\mid R \rangle$] $\langle S\mid R \rangle$ is a presentation of $G$. $S$ is a set of generators of $G$, while $R$ is a set of relations $G$ should satisfy. See Definition~\ref{def:presentation}.
\item[$1$] the natural number ``1'', identity element of a group, or the trivial group. The meaning should be clear in different contexts with no confusion.
\item[$\z_n$] the integers modulo $n$ $\cong$ the cyclic group of order $n$.
\item[$S_n$] the symmetric group of degree $n$, consisting of all permutations on $n$ points.
\item[$D_{2n}$] the dihedral group of order $2n$.
\item[$\f_q$] the finite field of $q$ elements.
\item[$\z_n^\times$, $\f_q^\times$] the multiplicative group of units of $\z_n$, and of $\f_q$, both consisting of all invertible elements under multiplication. $\f_q^\times$ = all nonzero elements of $\f_q$.
\item[$GL(n,q)$] the general linear group of degree $n$ on $\f_{q}$, which consists of all invertible $n \times n$ matrices with entries from $\f_q$. The identity element for $GL(n,q)$ is usually denoted by $I$ = identity matrix. $|GL(n,q)| = (q^n-1)(q^n-q)(q^n-q^2)\cdots(q^n-q^{n-1})$.
\item[$V_q$] the center of $GL(n,q)$, consisting of the collection of matrices that commute with every matrix in $GL(n,q)$ = all nonzero scalar matrices = $\{\alpha I:\alpha \in\f_q^\times\}$.
\item[$PGL(n,q)$] the projective general linear group = $GL(n,q)/V_q$. $|PGL(n,q)| = |GL(n,q)|/|V_q| = |GL(n,q)|/(q-1)$. In other words, it is the group of all invertible $n\times n$ matrices with entries from $\f_q$, where matrices that are proportional are considered the same group element.
\item[$SL(2,q)$] the special linear group = all matrices in $GL(2,q)$ with determinant 1. $|SL(2,q)|=|PGL(2,q)|$.
\item[$PSL(2,q)$] the projective special linear group = $SL(2,q)/\langle-I\rangle$. $|PSL(2,q)|=|SL(2,q)|/2$.
\end{description}

To simplify expressions in later sections, let $\K_n \triangleq \{0,1,\dots,n-2\}$ for integers $n \geq 2$.

\section{Ingleton Violation: Computer Search and Some Conditions}
\label{section:conditions}

Since the Ingleton inequality \eqref{equation:Ingletong} involves four subgroups of a finite group and their various intersections, designing a small admissible structure is very difficult without an existing example. So we use
computer programs to search for a small instance. Specifically, we use the GAP system\cite{GAP4} to search its ``Small Group'' library, which contains all finite groups of order less than
or equal to 2000, except those of 1024. We pick a group in this library (starting from the smallest, of course), find
all its subgroups, then test the Ingleton inequality for all
4-combinations of these subgroups. This is a tremendous task, as
there are already more than 1000 groups (up to isomorphism) of order less than or equal
to 100, each of which might have hundreds of subgroups, or even more.

It is therefore extremely critical to prune our search. In fact, we
used the following conditions to exclude groups or subgroups in the search,
each of which guarantees that Ingleton is satisfied.

\begin{condition}\label{condition:abelian}
$G$ is abelian.\cite{Chan-GroupCharEntropy}
\end{condition}
\begin{condition}\label{condition:normal}
$G_i \trianglelefteq G$, $\forall i$.\cite{Ingleton_Hom}
\end{condition}
\begin{condition}\label{condition:g1g2}
$G_1G_2 = G_2G_1$, or equivalently $G_1G_2 \leq G$.
\end{condition}
\begin{condition}\label{condition:gor1}
$G_i = 1$ or $G$, for some $i$.
\end{condition}
\begin{condition}\label{condition:distinct}
$G_i = G_j$ for some distinct $i$ and $j$.
\end{condition}
\begin{condition}\label{condition:g1intxg2}
$G_{12} = 1$.
\end{condition}
\begin{condition}\label{condition:subgroup}
$G_i \leq G_j$ for some distinct $i$ and $j$.
\end{condition}

Note that Condition~\ref{condition:normal} subsumes Condition~\ref{condition:abelian}, while Condition~\ref{condition:g1g2}
subsumes Condition~\ref{condition:normal}. Also Conditions~\ref{condition:gor1} and \ref{condition:distinct} are contained in Condition~\ref{condition:subgroup}. Nevertheless, we still list these more restrictive conditions as they are easier to check using computer programs. In addition, Conditions~\ref{condition:abelian}, \ref{condition:g1g2} and
\ref{condition:g1intxg2} are crucial in our program, as they appear in the outer loops and can save a large
amount of search work.

For the above reasons we only list the proofs for Conditions~\ref{condition:g1g2}, \ref{condition:g1intxg2} and \ref{condition:subgroup} below:
\begin{IEEEproof}[Proof~\ref{condition:g1g2}]
Construct random variables $X_i$'s from uniformly
distributed $\Lambda$ on $G$ as in Section~\ref{subsec:group}. As
$G_{1;2} \triangleq G_1G_2 \leq G$, we can similarly construct
random variable $Z = \Lambda G_{1;2}$. In fact $Z$ is a common information for $X_{1}$ and $X_{2}$: since $|G_{1;2}|=
|G_1||G_2|/|G_{12}|$,
\[ H(Z) = H(X_{1})+H(X_{2})-H(X_{1},X_{2}) = I(X_{1};X_{2}). \]
Also $H(Z|X_1) = H(Z|X_2) = 0$ as $G_1,G_2\leq G_{1;2}$. Thus Ingleton is satisfied by Lemma~\ref{lemma:comm-info}.
\end{IEEEproof}

In the proof above we used the group-entropy correspondence in Section~\ref{subsec:group} to translate the problem to the entropy domain. Henceforth, in order to show that a group satisfies Ingleton, we shall either prove \eqref{equation:Ingletong} directly, or equivalently prove \eqref{equation:Ingletonh} using this correspondence. Furthermore, observe that the Ingleton inequality has symmetries between subscripts 1
and 2 and between 3 and 4, i.e. if we interchange the subscripts 1
and 2, or 3 and 4, the inequality stays the same. Thus if we prove
conditions for some $i \in \{1,2\}$ and $j \in \{3,4\}$, we
automatically get conditions for all $(i,j) \in
\{1,2\}\times\{3,4\}$. So without loss of generality, we will
just prove conditions for the case $(i,j)=(1,3)$ when these symmetries apply.

\begin{IEEEproof}[Proof~\ref{condition:g1intxg2}]
Realize that \eqref{equation:Ingletonh} can be rewritten as
\begin{equation}\label{equation:Ingletonsubmod}
\delta_{13,14}+\delta_{23,24}+\delta_{134,234}-\delta_{123,124} \geq 0,
\end{equation}
where for $\emptyset \neq
\alpha,\beta \subseteq \n$,
\[ \delta_{\alpha,\beta} \triangleq h_\alpha + h_\beta - h_{\alpha\cap\beta} - h_{\alpha\cup\beta}. \]
For example, $\delta_{134,234} = h_{134} + h_{234}
- h_{34} - h_{1234}$. By submodularity of entropies, all $\delta_{\alpha,\beta} \geq
0$. If $G_{12} = 1$, then $\delta_{123,124} = 0$ and
(\ref{equation:Ingletonsubmod}) holds.
\end{IEEEproof}
\begin{IEEEproof}[Proof~\ref{condition:subgroup}]
$(i,j) = (1,2)$ implies Condition~\ref{condition:g1g2}. $(i,j) = (1,3)$ implies
$\delta_{123,124} = 0$ in (\ref{equation:Ingletonsubmod}). $(i,j) = (3,1)$
implies $\delta_{123,234} = 0$ and so $\delta_{123,234} \leq \delta_{12,24}$, which further transforms to
$\delta_{123,124} \leq \delta_{23,24}$, thus (\ref{equation:Ingletonsubmod}) holds. For $(i,j) = (3,4)$, \eqref{equation:Ingletong} becomes 
\[ |G_1||G_2||G_3||G_{123}||G_{124}| \geq |G_{12}||G_{13}||G_{14}||G_{23}||G_{24}|, \]
which is true as $G_2\geq G_{24}$ and by submodularity, $|G_1||G_{124}|\geq |G_{12}||G_{14}|$ and $|G_3||G_{123}|\geq |G_{13}||G_{23}|$.
\end{IEEEproof}

\section{The Smallest Violation Instance and the Group Presentation}
\label{section:presentation}

Using GAP we found the smallest group that violates Ingleton is $G =
S_5$, which has 60 sets of violating subgroups up to subscript symmetries. Further examination shows that these 60 sets of subgroups are in fact all conjugates of each other, thus are virtually the same in terms of group
structure. We list below
some information from GAP about one representative:\footnote{The permutations are written in cycle notation, e.g. $(1,2)(3,4,5)$ is the
permutation on the set $\{1,2,3,4,5\}$ that makes the following mapping: $1\mapsto 2,\ 2\mapsto 1,\ 3\mapsto 4,\ 4\mapsto 5,\ 5\mapsto 3$. Also GAP's convention for permutations is used throughout this paper, i.e. permutations are applied to an element from the right.}
\[ \begin{array}{lll}
G_1 = \big\langle (3,4,5), (1,2)(4,5) \big\rangle & \cong S_3 \cong D_6 & |G_1| = 6\\
G_2 = \big\langle (1,2,3,4,5), (1,4,3,5) \big\rangle & \cong \z_5 \rtimes
\z_4 & |G_2| = 20\\
G_3 = \big\langle (2,3), (1,3,4,2) \big\rangle & \cong D_8 & |G_3| = 8\\
G_4 = \big\langle (2,4), (1,2,5,4) \big\rangle & \cong D_8 & |G_4| = 8\\
G_{12} = \big\langle (1,2)(3,5) \big\rangle & \cong \z_2 & |G_{12}| = 2\\
G_{13} = \big\langle (1,2)(3,4) \big\rangle & \cong \z_2 & |G_{13}| = 2\\
G_{14} = \big\langle (1,2)(4,5) \big\rangle & \cong \z_2 & |G_{14}| = 2\\
G_{23} = \big\langle (1,3,4,2) \big\rangle & \cong \z_4 & |G_{23}| = 4\\
G_{24} = \big\langle (1,2,5,4) \big\rangle & \cong \z_4 & |G_{24}| = 4\\
G_{34} = 1 & & |G_{34}| = 1\\
G_{123} = 1 & & |G_{123}| = 1\\
G_{124} = 1 & & |G_{124}| = 1.\\
\end{array} \]
Simple calculation shows that
\[ |G_1||G_2||G_{34}||G_{123}||G_{124}| = 120 < 128 =
|G_{12}||G_{13}||G_{14}||G_{23}||G_{24}|, \]
so Ingleton is violated.
Also we can check that $G_1$--$G_4$ indeed generate $G$.

To illustrate the structure of these subgroups, we use the group
cycle graph. See Fig.~\ref{figure:flowerstructure}, where the
dash-dotted lines denote the pairwise intersections of subgroups
excluding identity.
\begin{figure}[!t]
\centering
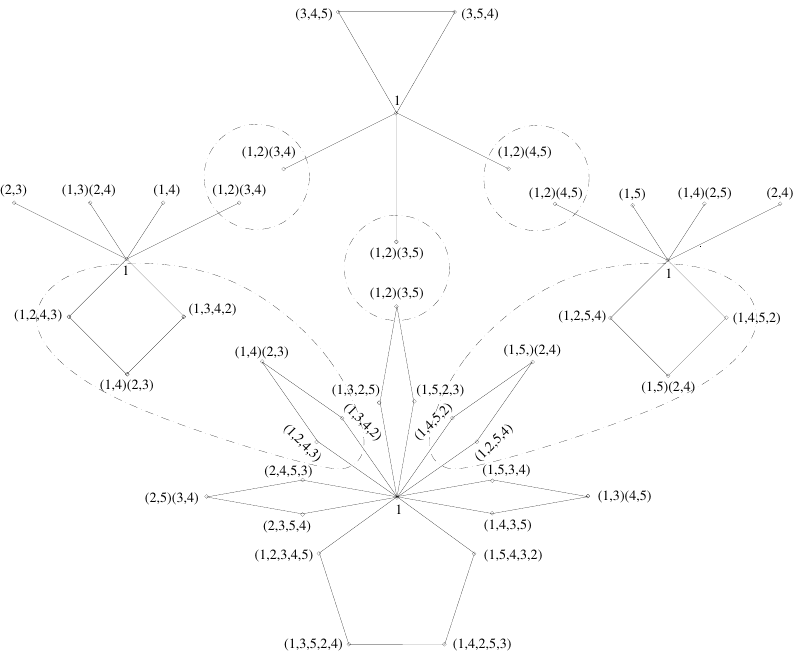
\caption{The cycle graph of the Ingleton violating subgroups of $S_5$}
\label{figure:flowerstructure}
\end{figure}
From the cycle graph we can obtain more structural information which
GAP does not show us directly. First, not only is $G_2$ a semidirect
product of two cyclic groups $\big\langle (1,2,3,4,5)\big\rangle \cong \z_5$
and $\big\langle (1,4,3,5)\big\rangle \cong \z_4$, but also $\big(G_2\setminus\big\langle (1,2,3,4,5)\big\rangle\big)
\bigcup\{1\}$ is the union of subgroups which are all isomorphic to (in fact, conjugate to)
$\big\langle (1,4,3,5)\big\rangle$ and have trivial pairwise intersections.
We say $G_2$ has a ``flower'' structure in this case. Second,
$G_4$ is the conjugate of $G_3$ by $(3,4,5)$. In
particular, there is a conjugacy relation between the order-4 generators of $G_{3}$ and $G_{4}$:
$(1,3,4,2)^{(3,4,5)} = (1,4,5,2) = (1,2,5,4)^{-1}$.

In order to generalize these subgroups to a family of violations,
we seek a parameterized group presentation for $G$ that retains the above structures. Although these group presentations are abstract, each of them can be input to GAP to yield an isomorphic concrete group, and Ingleton inequality can be checked against the corresponding subgroups. Observing that $|G_{23}|$ and
$|G_{24}|$ (both equal to 4) contribute most to the right-hand side ($RHS$)
of (\ref{equation:Ingletong}), we may try to let the ``petals'' of
$G_2$ (conjugates of $\big\langle (1,4,3,5)\big\rangle$) grow while keeping
other structures fixed.\footnote{This approach is a little conservative, but it is
the only successful extension according to our GAP trials. For
example, one may try to expand $G_1$ at the same time, but the
structures of $G_3$ and $G_4$ usually collapse.} In the rest of this section, we start from a presentation of $G_{2}$ and then extend it to the whole group $G$.

Let us first define a presentation of a group. For a precise definition one needs to introduce the concept of free groups, which we will skip. The interested readers may consult abstract algebra textbooks, e.g. \cite{absAlg,presGroup}. Here we only give an informal but useful definition.

\begin{definition}[Group Presentation]
\label{def:presentation}
A set $S$ of \emph{generators} of a group $G$ is a subset of $G$, such that every group element can be written as a finite product of elements of $S$ and their inverses. An equation satisfied in $G$ involving only $S\cup\{1\}$ is called a \emph{relation} in $G$ among $S$. Let $R$ be a set of such relations. We say $G$ has a \emph{presentation}
\[ \langle\, S \mid R\, \rangle \]
if $G$ is the largest (``freest'') group generated by $S$ subject only to the relations $R$. (Formally, the group $G$ is said to have the above presentation if it is isomorphic to the quotient of a free group $F$ on $S$ by the normal subgroup of $F$ generated by the relations $R$.)
\end{definition}

For example, consider a presentation $\langle x\mid x^{n}=1 \rangle$. Any group generated by $x$ contains only the powers of $x$, but by the relation $x^{n}=1$ the order of such a group cannot exceed $n$. Among these groups the cyclic group $\z_{n}$ has the maximum order, hence has the above group presentation.

\subsection{Presentation of $G_{2}$}
\label{subsec:presentation-g2}

Let $G_2$ be generated by two elements $a$ and $b$, with a normal subgroup $N = \langle
a\rangle \cong\z_n$ and another subgroup $H = \langle b\rangle
\cong\z_m$, for some integers $m,n$. This gives
us a presentation
\begin{equation}\label{equation:presentation_sdprod}
G_2 = \big\langle a,b\,\big|\,a^n=b^m=1,\, a^b=a^s\big\rangle
\end{equation}
for some $0<s<n$. In order to violate Ingleton as much as possible,
we may wish for $n$ to be small while $m$ is large. However, the flower
structure of $G_2$ may limit the choices of $n$ and $m$. First of
all, for this presentation to be a semidirect product, we need $s^m\equiv 1\pmod n$
(see \cite[Sec 5.4]{presGroup}), i.e., 
\begin{equation}\label{equation:s-order-m}
s\in \z_n^\times,\quad |s|\,\big|\,m,
\end{equation}
where $|s|$ denotes the order of $s$ in the multiplicative group $\z_n^\times$. As a consequence, $|G_2| = mn$, $H\bigcap N = 1$, and by the relations in \eqref{equation:presentation_sdprod} we also have
\begin{equation}\label{equation:ai-bk-ai-sk}
(a^i)^{b^{k}}=a^{is^k},\quad \forall i,k \in \z.
\end{equation}
Moreover, we need $(G_2\setminus N) \bigcup\{1\}$
to be the union of subgroups which are all isomorphic to $H$ with
trivial pairwise intersections.

One possible way to achieve this is to restrict $H^{g_1}\bigcap H^{g_2}
= 1,\ \forall g_1\neq g_2\in N$, as in our original example.
This is equivalent to $H^g\bigcap H = 1$, $\forall g\in
N\setminus\{1\}$. If this is the case, then there will be $|N|=n$
``petals'' of size $m$ in $G_2$, and the total number of nonidentity
elements will equal $n(m-1) = nm - n = |G_2\setminus N|$, and then indeed
the flower structure would be achieved.
Pick two nonidentity elements $h_1 = b^l\in H$, $h_2 =
(b^k)^{a^i}\in H^{a^i}$ for some $0<k,l<m$ and some $0<i<n$. Then
\[ h_1 = h_2\ \Leftrightarrow\  a^{-i}b^ka^i=b^l\  \Leftrightarrow\
a^{-i}(a^i)^{b^{-k}}b^k=b^l\  \Leftrightarrow\
a^{-i}a^{is^{-k}}=b^{l-k}\  \Leftrightarrow\  a^{(s^{-k}-1)i}=b^{l-k}. \]
In the last equation, $LHS\in N$ and $RHS\in H$. But $H\bigcap N = 1$ forces that $a^{(s^{-k}-1)i}=b^{l-k}=1$,
i.e. $l=k$ and $n\,\big|\,(s^{-k}-1)i$.

To guarantee that $H^{a^i}\bigcap H = 1$
, we must have $m\leq |s|$. Otherwise if we let
$0<k=|s|<m$, then $s^{-k} \equiv 1 \pmod n$ and so $n\,\big|\,(s^{-k}-1)i$ is satisfied. This means that by choosing $k=l= |s|$, we have found a nonidentity element $h_2 = (b^k)^{a^i} = b^l = h_1$ in $H^{a^i}\bigcap H$. Therefore $m\leq |s|$ and as $|s|\,\big|\,m$ by \eqref{equation:s-order-m}, $m=|s|$. In particular, $m\leq |\z_n^\times| \leq n-1$.

For $m$ to be as large as possible, $s$ should be a primitive root
modulo $n$, which makes $m=|\z_n^\times|$. Pick $n=p$ for some prime $p$, then $m = |\z_p^\times| = p-1$ achieves the upper bound $m \leq n-1$. Also in this case, if $0<k<m=|s|$ and $0<i<n=p$, then
$n\,\big|\,(s^{-k}-1)i$ requires $p\,\big|\,i$ or $p\,\big|\,(s^{-k}-1)$. Since $p>i$, the latter must be true, which implies that $|s|\,\big|\,k$. But this is a contradiction since $0<k<|s|$. So indeed
we have $H^{g}\bigcap H = 1$, $\forall g\in N$, and the flower structure
is realized. Furthermore, to make $H$ nontrivial, we need $p>2$. Thus with such a choice of parameters, the presentation of $G_2$ becomes
\begin{equation}\label{equation:presentationg2}
G_2 = \big\langle a,b\,\big|\,a^p=b^{p-1}=1,\, a^b=a^s\big\rangle,
\end{equation}
where $p>2$ is a prime and $s$ is a primitive root modulo $p$.

\subsection{Presentation of $G$}

The next step is to extend the presentation \eqref{equation:presentationg2} to the whole group $G$
generated by $G_1$--$G_4$, with the structure in Fig.~\ref{figure:flowerstructure}. Consider the dihedral groups $G_3$ and
$G_4$. The subgroups of rotations are just $H^{a_3}$ and $H^{a_4}$,
respectively, for some $a_3=a^{k_3},a_4=a^{k_4}\in N$. Also $G_3$
and $G_4$ each shares one element of reflection with the dihedral
group $G_1$, while the remaining reflection of $G_1$ is just
$(b^\fp)^{a_1}$ in $G_2$, for some $a_1=a^{k_1}\in N$.
Thus if we can determine the generator of the subgroup of rotations
of $G_1$, then all elements of $G_1$--$G_4$ are determined. In
other words, if we introduce an element $c$ as the generator of
rotations of $G_1$, then all elements from $G_1$--$G_4$ can be
express as products of $a,b,c$ and their inverses. Define
\begin{equation}\label{equation:b1b3b4}
b_1 = (b^\fp)^{a^{k_1}},\quad b_3 = b^{a^{k_3}},\quad
b_4 = b^{a^{k_4}}
\end{equation}
for some integers $k_1,k_3,k_4$. If in Fig.~\ref{figure:flowerstructure} we let $a,b,c,b_1,b_3,b_4$ correspond with the elements specified in Table~\ref{table:grp-elm-corrsp},
\begin{table}[!t]
\caption{Correspondence of group elements}
\label{table:grp-elm-corrsp}
\centering
\begin{tabular}{c|c|c|c|c|c}
\hline
$a$ & $b$ & $c$ & $b_{1}$ & $b_{3}$ & $b_{4}$\\
\hline
$(1,2,3,4,5)$ & $(1,4,3,5)$ & $(3,4,5)$ & $(1,2)(3,5)$ & $(1,3,4,2)$ & $(1,2,5,4)$\\
\hline
\end{tabular}
\end{table}
then the subgroups and the whole group in our presentation should be
\begin{equation}\label{equation:presentationsubgrps}
G_1 = \big\langle c,b_1\big\rangle,\quad G_2 = \big\langle a,b\big\rangle,\quad G_3
= \big\langle b_1c^2,b_3\big\rangle,\quad G_4 = \big\langle
b_1c,b_4\big\rangle,\quad G = \big\langle a,b,c\big\rangle.
\end{equation}
As $G_1\cong D_6$, we should have the relation $c^3=(cb_1)^2=1$. Furthermore, for
$G_3$ and $G_4$ to be dihedral groups, we need
$(b_3 \cdot b_1c^2)^2=(b_4 \cdot b_1c)^2=1$.

At this point we can try to plug in the presentation with these relations to GAP to find a concrete group. But still there are too many parameters to choose, especially when $p$ is large, the choices of $k_{1},k_{3},k_{4}$ are numerous. Also for a fixed $p$ not many such combinations yield successful Ingleton violations, according to our GAP trials. Therefore we need to ultilize more structural information from Fig.~\ref{figure:flowerstructure} to obtain more restrictions on $k_{1},k_{3}$ and $k_{4}$.

Observe that in the original violation, $G_4$ is the conjugate of $G_3$ by $(3,4,5)$,
and $(1,3,4,2)^{(3,4,5)} = (1,2,5,4)^{-1}$. In our presentation this translates to
$b_3^c = b_4^{-1}$, according to Table~\ref{table:grp-elm-corrsp}. With this new relation, we claim that $(b_3 \cdot b_1c^2)^2=(b_4 \cdot b_1c)^2=1$ is satisfied if and only if
$k_3 - k_1 \equiv k_1 - k_4 \pmod p$. In fact, as $|b_1| = 2$,
$c^3=(cb_1)^2=1$, we have $cb_1 = b_1c^2$ and $b_1c = c^2b_1$. Using these relations we can establish
the following equalities:
\[ (b_3 \cdot b_1c^2)^2 = b_3b_1c^{-1}b_3cb_1 = b_3b_1b_4^{-1}b_1, \]
\[ (b_4 \cdot b_1c)^2 = b_4b_1cb_4c^{-1}b_1 = b_4b_1b_3^{-1}b_1 = \big((b_3b_1b_4^{-1}b_1)^{-1}\big)^{b_1}. \]
So $(b_3 \cdot b_1c^2)^2=1$ if and only if $(b_4 \cdot b_1c)^2=1$. Using \eqref{equation:ai-bk-ai-sk} and the fact that $b^{\fp} = (b^{\fp})^{-1}$ and plugging (\ref{equation:b1b3b4}) in, we have
\begin{align*}
b_3b_1b_4^{-1}b_1 &= b^{a^{k_{3}}}(b^{\fp})^{a^{k_{1}}}(b^{-1})^{a^{k_{4}}}(b^{\fp})^{a^{k_{1}}} \\
& = a^{-k_{3}}ba^{k_{3}-k_{1}}b^{\fp}a^{k_{1}-k_{4}}b^{-1}a^{k_{4}-k_{1}}b^{\fp}a^{k_{1}} \\
& = a^{-k_{3}}\cdot ba^{k_{3}-k_{1}}b^{-1}\cdot b^{\fp}\cdot ba^{k_{1}-k_{4}}b^{-1}\cdot a^{k_{4}-k_{1}}b^{\fp}a^{k_{1}} \\
& = a^{-k_{3}}\cdot a^{(k_{3}-k_{1})s^{-1}}\cdot b^{\fp}\cdot a^{(k_{1}-k_{4})s^{-1}}\cdot a^{k_{4}-k_{1}}b^{\fp}a^{k_{1}} \\
& = a^{(k_{3}-k_{1})s^{-1}-k_{3}}\cdot (b^{\fp})^{-1}a^{(k_{1}-k_{4})(s^{-1}-1)}b^{\fp}\cdot a^{k_{1}} \\
& = a^{(k_{3}-k_{1})s^{-1}-k_{3}}\cdot a^{(k_{1}-k_{4})(s^{-1}-1)s^{(p-1)/2}}\cdot a^{k_{1}} \\
& = a^{[(k_3-k_1)+(k_1-k_4)s^{(p-1)/2}](s^{-1}-1)}.
\end{align*}
Since $s$ is a primitive root modulo $p$, $|s^{(p-1)/2}|=2$. As
$\z_p^\times$ is cyclic of an even order $p-1$, it is clear that
there is a unique element of order $2$. But $-1$ has order 2 in $\z_p^\times$, so $s^{(p-1)/2} \equiv -1 \pmod p$ and
\[ (b_3 \cdot b_1c^2)^2 = b_3b_1b_4^{-1}b_1 = a^{[(k_3-k_1)-(k_1-k_4)](s^{-1}-1)}. \]
Now $p\nmid(s^{-1}-1)$ as $s\neq 1$, which implies
\[ (b_3 \cdot b_1c^2)^2 = 1\ \Leftrightarrow\ p\,\big|\,[(k_3-k_1)-(k_1-k_4)]\
\Leftrightarrow\ k_3 - k_1 \equiv k_1 - k_4 \pmod p. \]
This condition on $k_{1},k_{3}$ and $k_{4}$ tells us that the petals $G_{23}$ and $G_{24}$ of $G_2$ should be symmetric (modulo $p$) w.r.t. $G_{12}$, i.e. $G_{23}$, $G_{12}$ and $G_{24}$ should be equally spaced.\footnote{With this symmetry it is very easy for GAP to produce the desired structures, even with arbitrary choices of $k_1$ and $k_3$.}

In sum, our analysis leads to the following presentation:
\begin{equation}\label{equation:presentationg}
G = \big\langle a,b,c\,\big|\,a^p=b^{p-1}=c^3=1,\, a^b=a^s,\,
(cb_1)^2=b_3^cb_4=1\big\rangle
\end{equation}
where $p$ is an odd prime, $s$ is a primitive root modulo $p$, $k_3
- k_1 \equiv k_1 - k_4 \pmod p$. If our extension of the subgroup
structures succeeds, then the orders of subgroups and intersections
would be: $|G_1| = 6$, $|G_2| = p(p-1)$, $|G_3| = |G_4| = 2(p-1)$,
$|G_{12}| = |G_{13}| = |G_{14}| = 2$, $|G_{23}| = |G_{24}| = p-1$,
$|G_{34}| = |G_{123}| = |G_{124}| =1$. Hence $LHS$ of
(\ref{equation:Ingletong}) $=6p(p-1)$ while $RHS = 8(p-1)^2$, and so when
$p\geq5$, Ingleton should be violated.

\section{Explicit Violation Construction with $PGL(2,p)$ and $PGL(2,q)$}
\label{section:pgl}

Feeding the above presentation to GAP, we find that for $p = 5,7,\dots,23$ the outcome is a finite group that violates the Ingleton inequality.\footnote{The capability of the testing program is primarily limited by hardware. When $p$ is too large the program runs out of memory.} Moreover, with GAP we verified for the first few primes (up to $p = 11$) that this group is isomorphic to the projective general linear group $PGL(2,p)$. This leads us to conjecturing that $PGL(2,p)$ is a family of Ingleton-violating groups. In fact, with an explicit identification of the generators in \eqref{equation:presentationg} with matrices in $PGL(2,p)$, we prove that $PGL(2,p)$ is indeed a family of Ingleton-violating groups for primes $p\geq 5$, by directly constructing their violating subgroups in \eqref{equation:presentationsubgrps} in the form of matrices. These matrix subgroups all have clear interpretations. Furthermore, once we have the formats of these subroups, we extend them to the Ingleton-violating family $PGL(2,q)$ for all finite field order $q\geq 5$.

\subsection{The Family $PGL(2,p)$}
\label{subsec:pgl(2,p)}

First we introduce some necessary notations. Let $p$ be an odd prime. For $A\in GL(2,p)$, let $\A$ denote the
left coset of $A$ in $GL(2,p)$ with respect to the center $V_p = \{\alpha
I:\alpha \in\f_p^\times\}$. Thus $\A = \B$ if and only if each entry
of $A$ is a nonzero constant multiple of the corresponding entry of
$B$. $A^T$ denotes the transpose of $A$ as usual.
We denote the elements of $\f_p$ by ordinary integers, but the
addition and multiplication, as well as equality, are modulo $p$. Furthermore, $-k$ and $k^{-1}$ denotes the additive and multiplicative inverses of $k$ in $\f_p$ respectively. If $s\in \f_p$, and $A$ has multiplicative order $p$, then $A^s$ simply indicates the $s$-th power of $A$, where $s$ is viewed as an integer.

We start by identifying the generators in $PGL(2,p)$ that correspond to presentation~\eqref{equation:presentationg}. Consider the following matrices in $GL(2,p)$:
\[ A = \begin{bmatrix} 1 & 0 \\ 1 & 1 \end{bmatrix},\qquad
B = \begin{bmatrix} 1 & 0 \\ 0 & t \end{bmatrix},\qquad
C = \begin{bmatrix} 0 & 1 \\ -1 & -1 \end{bmatrix} \]
where $t$ is a primitive root modulo $p$, i.e. a generator of
$\f_p^\times$. Our guess is that $\A,\B,\C$ correspond to the generators
$a,b,c$ in (\ref{equation:presentationg}) respectively. The powers
of these matrices are:
\[ A^k = \begin{bmatrix} 1 & 0 \\ k & 1 \end{bmatrix},\quad
B^k = \begin{bmatrix} 1 & 0 \\ 0 & t^k \end{bmatrix},\quad
C^2 = \begin{bmatrix} -1 & -1 \\ 1 & 0 \end{bmatrix},\quad
C^3 = I \]
for any integer $k$. Thus 
$\big|\A\big| = p$, $\big|\B\big| = p-1$, and $\big|\C\big| = 3$. Also,
\[ A^B = B^{-1}AB = \begin{bmatrix} 1 & 0 \\ t^{-1} & 1 \end{bmatrix} = A^s, \]
where $s = t^{-1}$ is also a primitive root modulo $p$. So $\A^{\B}
= \A^s$. Next we let
\[ B_1 = (B^\fp)^{A^{k_1}} =
\begin{bmatrix} 1 & 0 \\ -k_1 & 1 \end{bmatrix}
\begin{bmatrix} 1 & 0 \\ 0 & -1 \end{bmatrix}
\begin{bmatrix} 1 & 0 \\ k_1 & 1 \end{bmatrix} =
\begin{bmatrix} 1 & 0 \\ -2k_1 & -1 \end{bmatrix}, \]
where we calculated $t^\fp = -1$ as it is the unique
element of order $2$ in $\f_p^\times$. Now check
\[ CB_1 = \begin{bmatrix} -2k_1 & -1 \\ 2k_1-1 & 1 \end{bmatrix},\quad
(CB_1)^2 = \begin{bmatrix} 4k_1^2-2k_1+1 & 2k_1-1 \\ -(2k_1-1)^2 & 2-2k_1 \end{bmatrix}. \]
Thus if we want $\big(\C\Ba\big)^2 = \I$, $k_1$ must be $2^{-1}=\mfp$. In this case,
\[ B_1 =\begin{bmatrix} 1 & 0 \\ -1 & -1 \end{bmatrix},\quad
CB_1 = \begin{bmatrix} -1 & -1 \\ 0 & 1 \end{bmatrix},\quad
\big(\C\Ba\big)^2 = \I. \]
Let $B_3=B^{A^{k_3}}$, $B_4=B^{A^{k_4}}$. As $k_3 - k_1 = k_1 - k_4$, we have $k_3=1-k_4$.
\[  B^{A^k} = \begin{bmatrix} 1 & 0 \\ k(t-1) & t \end{bmatrix},\quad
B_3C\cdot B_4 = \begin{bmatrix} 0 & 1 \\ -t & k_3(t-1)-t \end{bmatrix}
\begin{bmatrix} 1 & 0 \\ k_4(t-1) & t \end{bmatrix}, \]
whose $(1,1)$-entry is $k_4(t-1)$. If we want
$\Bc^{\C}\Bd=\I$, i.e., $\Bc\C\Bd=\C$, $k_4$ must be
$0$ since the $(1,1)$-entry of $C$ is $0$ and $t\neq 1$. So $k_3=1-k_4=1$,
\[ B_3 = \begin{bmatrix} 1 & 0 \\ t-1 & t \end{bmatrix},\quad
B_4 = \begin{bmatrix} 1 & 0 \\ 0 & t \end{bmatrix} = B,\quad
\Bc\C\Bd = \overline{\begin{bmatrix} 0 & 1 \\ -t & -1 \end{bmatrix}
\begin{bmatrix} 1 & 0 \\ 0 & t \end{bmatrix}} =
\overline{\begin{bmatrix} 0 & t \\ -t & -t \end{bmatrix}} = \C. \]

So far for $\A,\B,\C$ we have verified all the relations in
(\ref{equation:presentationg}). We can also prove that they are
actually a set of generators for $PGL(2,p)$. Observe that each
matrix in $GL(2,p)$ can be written as a product of some elementary matrices,
which are
\[ \begin{bmatrix} 1 & 0 \\ \alpha & 1 \end{bmatrix},\quad
\begin{bmatrix} 1 & \beta \\ 0 & 1 \end{bmatrix},\quad
\begin{bmatrix} 1 & 0 \\ 0 & t^i \end{bmatrix},\quad
\begin{bmatrix} t^j & 0 \\ 0 & 1 \end{bmatrix} \]
where $\alpha,\beta\in\f_p$ and $i,j\in\K_p$. They are generated by $A,A^T,B$ and $t^{-1}B$ respectively. So $PGL(2,p)$ is generated
by $\A,\overline{A^T}$ and $\B$. Now check
\[ B_1C = \begin{bmatrix} 0 & 1 \\ 1 & 0 \end{bmatrix},\quad
A^{B_1C} = \begin{bmatrix} 1 & 1 \\ 0 & 1 \end{bmatrix} = A^T, \]
thus $\A,\B$ and $\C$ generate $PGL(2,p)$. Hence setting $s = t^{-1}$,
$k_1 = \mfp$, $k_3 = 1$, $k_4 = 0$, we see that $PGL(2,p)$ is a quotient of the group $G$ in \eqref{equation:presentationg}, whose generators $\A,\B$ and $\C$ correspond precisely to the
generators $a,b$ and $c$ of $G$.

\begin{remark}
Note that we have not proved that (\ref{equation:presentationg}) is
a presentation of $PGL(2,p)$. To do that, one must show
that the order of the group generated by $a,b,c$ in (\ref{equation:presentationg}) is no
more than $|PGL(2,p)| = (p-1)p(p+1)$, which we have not yet been able to prove.
However, identifying possible corresponding generators still gives
us a way to explicitly construct the subgroups to violate Ingleton.
\end{remark}

Now we can write out the subgroups in $PGL(2,p)$ corresponding to
subgroups in (\ref{equation:presentationsubgrps}).

$G_1 = \big\langle\C,\Ba\big\rangle$. Note that $\big|\C\big|=3$, $\big|\Ba\big|=2$, and
$\big(\C\Ba\big)^2 = \I$, so $\C\Ba = \Ba\big(\C\big)^2$ and $G_1$ has at
most $6$ elements $\{\big(\Ba\big)^i\big(\C\big)^j: 0\leq i <2,\; 0\leq j < 3\}$.
Calculating these elements we can see $|G_1|=6$ exactly and thus indeed $G_1 \cong D_6 \cong S_3$:
\[ G_1 = \left\{\I,\quad
\overline{\begin{bmatrix} 0 & 1 \\ -1 & -1 \end{bmatrix}},\quad
\overline{\begin{bmatrix} -1 & -1 \\ 1 & 0 \end{bmatrix}},\quad
\overline{\begin{bmatrix} 1 & 0 \\ -1 & -1 \end{bmatrix}},\quad
\overline{\begin{bmatrix} 0 & 1 \\ 1 & 0\end{bmatrix}},\quad
\overline{\begin{bmatrix} -1 & -1 \\ 0 & 1 \end{bmatrix}} \right\}. \]

$G_2 = \big\langle\A,\B\big\rangle$. We claim that $G_2$ is the subgroup of lower triangular
matrices\footnote{We would end up with upper triangular matrices for $G_2$ if $A^T$ were used in place of $A$. But the two resulting groups are actually conjugate to each other, e.g. consider conjugating by $B_{1}C$.} in $GL(2,p)$ modulo $V_p$, i.e.,
\[ G_2 = \left\{\left.\overline{\begin{bmatrix} 1 & 0 \\ \alpha & \beta \end{bmatrix}}\ \right|
\alpha\in\f_p,\ \beta\in\f_p^\times\right\}. \]
As $A,B$ are lower triangular, any element in $G_2$ is a lower
triangular matrix modulo $V_p$. On the other hand, $\forall
\alpha\in\f_p,\;\beta\in\f_p^\times$, then $\beta = t^l$ for some integer
$l$. So
\[ \begin{bmatrix} 1 & 0 \\ \alpha & \beta \end{bmatrix} =
A^{\alpha}B^l \Rightarrow \overline{\begin{bmatrix} 1 & 0 \\ \alpha & \beta\end{bmatrix}} = \A^{\alpha}\B^l \in G_2. \]
Thus $|G_2| = p(p-1)$ and $G_2$ has presentation~\eqref{equation:presentationg2}. Therefore, as proved in Section~\ref{subsec:presentation-g2}, $G_2 \cong \z_p\rtimes\z_{p-1}$ and it achieves the desired flower structure.

$G_3 = \big\langle\Ba\big(\C\big)^2,\Bc\big\rangle = \big\langle\C\Ba,\Bc\big\rangle$. Note
that $\big|\C\Ba\big|=2$, $\big|\Bc\big|=\big|\B\big|=p-1$, also
\[ B_3^k  = \begin{bmatrix} 1 & 0 \\ t^k-1 & t^k \end{bmatrix},\quad
B_3^{-1} = \begin{bmatrix} 1 & 0 \\ t^{-1}-1 & t^{-1} \end{bmatrix}, \]
\[ \Bc\cdot\C\Ba =
\overline{\begin{bmatrix} -1 & -1 \\ 1-t & 1 \end{bmatrix}} =
\overline{\begin{bmatrix} -t^{-1} & -t^{-1} \\ t^{-1}-1 & t^{-1} \end{bmatrix}}
= \C\Ba\big(\Bc\big)^{-1}, \]
so $G_3$ has at most $2(p-1)$ elements $\{\big(\C\Ba\big)^i\big(\Bc\big)^j: 0\leq i
<2,\; 0\leq j < p-1\}$. Calculating these elements we can see
$|G_3|=2(p-1)$ exactly and so $G_3 \cong D_{2(p-1)}$:
\[ G_3 = \left\{\left.\big(\Bc\big)^k =
\overline{\begin{bmatrix} 1 & 0 \\ t^k-1 & t^k \end{bmatrix}},\quad
\C\Ba\big(\Bc\big)^k = \overline{\begin{bmatrix} -1 & -1 \\ 1-t^{-k} & 1 \end{bmatrix}}\ \right| k\in\K_p \right\}. \]

$G_4 = \big\langle\Ba\C,\Bd\big\rangle$. Note that $\big|\Ba\C\big|=2$, $\big|\Bd\big|=\big|\B\big|=p-1$. Moreover,
\[ \Bd\cdot\Ba\C = \overline{\begin{bmatrix} 0 & 1 \\ t & 0 \end{bmatrix}} =
\overline{\begin{bmatrix} 0 & t^{-1} \\ 1 & 0 \end{bmatrix}} =\Ba\C\big(\Bd\big)^{-1}, \]
so $G_4$ has at most $2(p-1)$ elements $\{\big(\Ba\C\big)^i\big(\Bd\big)^j: 0\leq i
<2,\; 0\leq j < p-1\}$. Calculating these elements we can see
$|G_4|=2(p-1)$ exactly and so $G_4 \cong D_{2(p-1)}$:
\[ G_4 = \left\{\left.\big(\Bd\big)^k = \overline{\begin{bmatrix} 1 & 0 \\ 0 & t^k \end{bmatrix}},\quad
\Ba\C\big(\Bd\big)^k = \overline{\begin{bmatrix} 0 & t^k \\ 1 & 0\end{bmatrix}}\ \right| k\in\K_p \right\}. \]
These are all diagonal and anti-diagonal matrices in $GL(2,p)$ modulo $V_p$. Note that we have already verified $\big(\Bc\big)^{\C}=\Bd^{-1}$, also $\big(\C\Ba\big)^{\C} = \Ba\C$, thus indeed $G_{4}=G_{3}^{\C}$ as in the original instance (Fig.~\ref{figure:flowerstructure}).

With all four subgroups explicitly written, we can easily write down
the intersections:
\[ G_{12} = \big\langle\Ba\big\rangle =
\left\{\I,\quad\overline{\begin{bmatrix} 1 & 0 \\ -1 & -1\end{bmatrix}}\right\} \cong \z_2,\quad
G_{13} = \big\langle\C\Ba\big\rangle =
\left\{\I,\quad\overline{\begin{bmatrix} -1 & -1 \\0 & 1\end{bmatrix}}\right\} \cong \z_2, \]
\[ G_{14} = \big\langle\Ba\C\big\rangle =
\left\{\I,\quad\overline{\begin{bmatrix} 0 & 1 \\ 1 & 0\end{bmatrix}}\right\} \cong \z_2,\quad
G_{23} = \big\langle\Bc\big\rangle =
\left\{\left.\overline{\begin{bmatrix} 1 & 0 \\ t^k-1 & t^k \end{bmatrix}}\ \right| k\in\K_p \right\} \cong \z_{p-1}, \]
\[ G_{24} = \big\langle\Bd\big\rangle =
\left\{\left.\overline{\begin{bmatrix} 1 & 0 \\ 0 & t^k \end{bmatrix}}\ \right| k\in\K_p \right\} \cong \z_{p-1},\quad
G_{34} = G_{123} = G_{124} = 1. \]
\[ |G_{12}|= |G_{13}| = |G_{14}|=2,\quad |G_{23}|=|G_{24}|=p-1. \]
So in (\ref{equation:Ingletong}), indeed $LHS =
|G_1||G_2||G_{34}||G_{123}||G_{124}| = 6p(p-1)$ and $RHS =
|G_{12}||G_{13}||G_{14}||G_{23}||G_{24}| = 8(p-1)^2$, hence $LHS - RHS =
2(p-1)(4-p)$. Thus Ingleton is violated when $p\geq 5$, and the subgroup
structures of $S_5\cong PGL(2,5)$ are exactly reproduced.

\subsection{The Family $PGL(2,q)$}
\label{subsec:pgl(2,q)}

With the explicit matrix forms of the Ingleton-violating subgroups, we can further extend the above violation to $PGL(2,q)$, for all finite field order $q\geq 5$. For a finite field $\f_q$, we know that $q=p^m$ for some prime $p$ (the characteristic of $\f_q$) and some integer $m$.
Since $\f_p$ is the prime subfield of $\f_q$, $GL(2,p)$ is a subgroup of $GL(2,q)$, which induces an isomorphic copy of  $PGL(2,p)$ as a subgroup of $PGL(2,q)$. Therefore, using the same subgroups of $PGL(2,p)$ as in the previous section, we obtain a trivial Ingleton violation in $PGL(2,q)$ whenever the characteristic $p\geq 5$. Nevertheless, by extending the interpretations of these subgroups to $PGL(2,q)$, we can obtain a more general (nontrivial) violation, for each finite field order $q\geq 5$.

In the field $\f_q$, we continue to use the ordinary integers with modular arithmetic to represent the prime subfield $\f_p$. With this convention, all the matrices and subgroups in Section~\ref{subsec:pgl(2,p)} are well defined\footnote{The only problem that may arise is when $p=2$, $B_1 = (B^\fp)^{A^{k_1}}$ is not well defined. But we can circumvent that by directly working with the final matrix form of $B_1$.}, although now the cosets are taken with respect to $V_q$ rather than $V_p$.
These subgroups constitute a trivial embedding of our previous violation in $PGL(2,q)$. However, in $PGL(2,q)$, the previous sets of generators do not guarantee that $G_2$ is the full subgroup of all lower triangular matrices, nor that $G_4$ contains all the diagonal and anti-diagonal matrices.

To preserve these interpretations of the subgroups, we need to make some adjustment to the generators of $G_{2}$. Redefine $t$ to be a primitive element of $\f_q$, i.e. $t$ generates $\f_q^\times$. Then $\big|\B\big|= q-1$. Also instead of a single $A$, we need to introduce more matrices to generate the subgroup $N \triangleq \left\{\left.\As{\alpha}\, \right| \alpha\in\f_q\right\}$, where for each $\alpha\in\f_q$ we define
\[ A_\alpha = \begin{bmatrix} 1 & 0 \\ \alpha & 1 \end{bmatrix}. \]
Clearly $A_\alpha A_\beta =A_{\alpha+\beta}$, and $A_\alpha^k = A_{k\alpha}$ for each integer $k$. Thus $\big|\As{\alpha}\big|=p$ for each $\alpha\in\f_q^\times$. Observe that $\f_q$ is an $m$-dimensional vector space over $\f_p$, let $\left(\xi_1, \xi_2,\ldots,\xi_m\right)$ be a basis. Then $\forall\alpha\in\f_q$, $\alpha = \sum_{i=1}^mk_i\xi_i$ for some $k_1,k_2,\ldots,k_m\in\f_p$ and $A_\alpha = \prod_{i=1}^mA_{\xi_i}^{k_i}$. Also $\big\langle \As{\xi_i} \big\rangle\bigcap \big\langle \As{\xi_j} \big\rangle = 1$ for distinct $i$ and $j$. Thus
\[ N = \big\langle\ \As{\xi_1},\ \As{\xi_2},\ \ldots,\ \As{\xi_m}\ \big\rangle \cong \big\langle\ \As{\xi_1}\ \big\rangle\times\big\langle\ \As{\xi_2}\ \big\rangle\times\ldots\times\big\langle\ \As{\xi_m}\ \big\rangle \cong \z_p^m. \]
Actually, $N$ is isomorphic to the additive group of the vector space $\f_q$ over $\f_p$ (Also see Section~\ref{subsec:embedding}). 

Let $G_2 = \big\langle\As{\xi_1},\As{\xi_2},\ldots,\As{\xi_m},\B\big\rangle = \big\langle N,\B\big\rangle$. Similar to the previous section, it is easy to show that now $G_2$ is indeed the subgroup of all lower triangular matrices modulo $V_q$. Furthermore, for any $\alpha\in\f_q$, we have $\As{\alpha}^{\B} = \As{t^{-1}\alpha}$, so $N \trianglelefteq G_2$ and $G_2 = NH$, where $H\triangleq\big\langle\B\big\rangle$. Also $N\bigcap H = 1$, thus $G_2 \cong N\rtimes H \cong \z_p^m\rtimes\z_{q-1}$. 
\begin{figure}[!t]
\centering
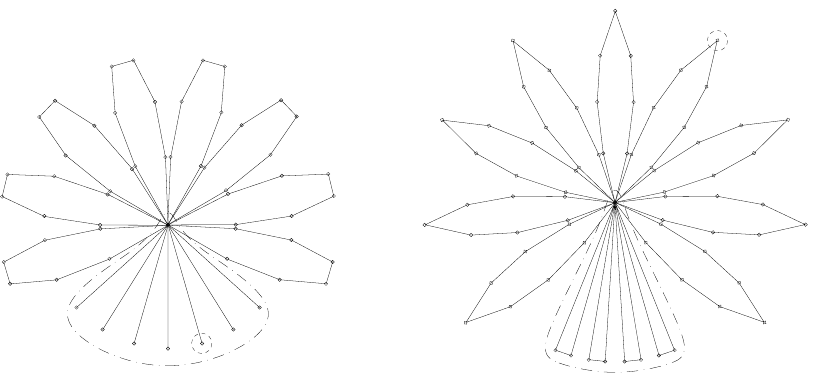
\caption{The generalized flower structures. The center point of each cycle graph denotes the identity element.}
\label{figure:genflowerstructure}
\end{figure}
Although in general $G_2$ does not have presentation~\eqref{equation:presentation_sdprod} or \eqref{equation:presentationg2} anymore since $N$ is not necessarily cyclic, we can prove that it does have a ``generalized flower structure'' when $q>2$, i.e. $(G_2\setminus N) \bigcup\{\I\}$ is the union of subgroups which are all isomorphic to $H$ with trivial pairwise intersections. Similar to the analysis of the $G_2$ in Section~\ref{subsec:presentation-g2}, it suffices to show that $H^{\As{\alpha}}\bigcap H = 1$, $\forall \As{\alpha}\in N\setminus\{\I\}$. But this is true since for each $\alpha\in\f_q^\times$ and some integers $k,l\in\K_q$,
\[ (\B^k)^{\As{\alpha}} = \B^l \iff \B^k\cdot\As{\alpha} = \As{\alpha}\cdot\B^l \iff
\overline{\begin{bmatrix} 1 & 0 \\ t^k\alpha & t^k \end{bmatrix}} = 
\overline{\begin{bmatrix} 1 & 0 \\ \alpha & t^l \end{bmatrix}} \iff k = l = 0.\]

Fig.~\ref{figure:genflowerstructure} shows two representative generalized flower structures of $G_2$, for $q=8$ and $q=9$. In each cycle graph of $G_2$, there are $|N|=q$ petals and one ``root system'' (encircled by the dash-dotted line), which is the normal subgroup $N$. Every petal is a conjugate of $H$ and has size $q-1$. 
Since $N$ has $q-1$ nonidentity elements, each having order $p$, the root system consists of $(q-1)/(p-1)$ trivially intersecting ``roots/tubers'', each of which is a $p$-cycle.
Note that when $m=1$, there is only one root/tuber, as in the original flower structure in Fig.~\ref{figure:flowerstructure}.

Now using the same matrices 
\[ C = \begin{bmatrix} 0 & 1 \\ -1 & -1 \end{bmatrix},\quad
B_1 = \begin{bmatrix} 1 & 0 \\ -1 & -1 \end{bmatrix},\quad
B_3 = B^{A_1} = \begin{bmatrix} 1 & 0 \\ t-1 & t \end{bmatrix},\quad
B_4 = B = \begin{bmatrix} 1 & 0 \\ 0 & t \end{bmatrix} \]
as in Section~\ref{subsec:pgl(2,p)} (except that $t$ now generates  $\f_q^\times$ instead of $\f_p^\times$), we write down the following subgroups:

$G_1 = \big\langle\C,\Ba\big\rangle \cong D_6 \cong S_3$. (Same as in Section~\ref{subsec:pgl(2,p)}.)

$G_2 = \big\langle\As{\xi_1},\As{\xi_2},\ldots,\As{\xi_m},\B\big\rangle = \big\langle N,\B\big\rangle \cong \z_p^m\rtimes\z_{q-1}$, which consists of all lower triangular matrices in $GL(2,q)$ modulo $V_q$.

$G_3 = \big\langle\Ba\big(\C\big)^2,\Bc\big\rangle = \big\langle\C\Ba,\Bc\big\rangle$. Now $\big|\Bc\big|=q-1$, and we still have $\Bc\cdot\C\Ba = \C\Ba\big(\Bc\big)^{-1}$, so
\[ G_3 = \left\{\left.\big(\Bc\big)^k =
\overline{\begin{bmatrix} 1 & 0 \\ t^k-1 & t^k \end{bmatrix}},\quad
\C\Ba\big(\Bc\big)^k = \overline{\begin{bmatrix} -1 & -1 \\ 1-t^{-k} & 1 \end{bmatrix}}\ \right| k\in\K_q \right\} \cong D_{2(q-1)}. \]

$G_4 = \big\langle\Ba\C,\Bd\big\rangle$. Now $\big|\Bd\big|=q-1$ and $\Bd\cdot\Ba\C = \Ba\C\big(\Bd\big)^{-1}$, so
\[ G_4 = \left\{\left.\big(\Bd\big)^k = \overline{\begin{bmatrix} 1 & 0 \\ 0 & t^k \end{bmatrix}},\quad
\Ba\C\big(\Bd\big)^k = \overline{\begin{bmatrix} 0 & t^k \\ 1 & 0\end{bmatrix}}\ \right| k\in\K_q \right\} \cong D_{2(q-1)}, \]
which comprises all diagonal and anti-diagonal matrices in $GL(2,q)$ modulo $V_q$.

Next we find the intersections: $G_{12} = \big\langle\Ba\big\rangle$, $G_{13} = \big\langle\C\Ba\big\rangle$, and
$G_{14} = \big\langle\Ba\C\big\rangle$, which are all isomorphic to $\z_2$;
$G_{23} = \big\langle\Bc\big\rangle$ and $G_{24} = \big\langle\Bd\big\rangle$, both of which are isomorphic to $\z_{q-1}$; and
$G_{34} = G_{123} = G_{124} = 1$.

The orders of the four subgroups are $|G_1|=6$, $|G_2| = q(q-1)$, $|G_3|=|G_4|=2(q-1)$, and for the intersections
$|G_{12}|= |G_{13}| = |G_{14}|=2$, $|G_{23}|=|G_{24}|=q-1$, $|G_{34}| = |G_{123}| = |G_{124}| = 1$. So in (\ref{equation:Ingletong}), $LHS = |G_1||G_2||G_{34}||G_{123}||G_{124}| = 6q(q-1)$, while $RHS =
|G_{12}||G_{13}||G_{14}||G_{23}||G_{24}| = 8(q-1)^2$. Thus $LHS - RHS =
2(q-1)(4-q)$ and Ingleton is violated when $q\geq 5$.

\begin{remark}
Depending on the characteristic $p$ of $\f_q$, the intersection $G_{12} = \big\langle\Ba\big\rangle$ might lie in either the petals or the roots of $G_2$, as depicted by the dashed circles in Fig.~\ref{figure:genflowerstructure}. If $p\neq2$, then $q$ is odd and $\Ba=(\B^\fq)^{\As{k_1}}$ where $k_1=2^{-1}=\mfp$, so $G_{12}$ is on the petal $H^{\As{k_1}}$. Whereas if $p=2$, then $-1=1$ and $\Ba=\As{1}\in N$, so $G_{12}$ becomes a root. Note that the patterns of the other intersections are not changed for different $q$'s.
\end{remark}

\begin{remark}
We can also show that $\As{\xi_1}, \As{\xi_2}, \ldots, \As{\xi_m}, \B$ and $\C$ generate $PGL(2,q)$, using the same argument as in the previous section. The only difference is that the elementary matrices of $GL(2,q)$ are now generated by $A_{\xi_1},A_{\xi_1}^T,\ldots,A_{\xi_m},A_{\xi_m}^T,B$ and $t^{-1}B$. But as $A_\alpha^{B_1C} = A_\alpha^T$, $\forall \alpha\in\f_q$, we see that $PGL(2,q)$ is indeed generated by the desired elements.
\end{remark}

In Section~\ref{section:Asch}, we will see that these subgroups have more fundamental interpretations in the framework of group actions and groups of Lie type: each subgroup is the stabilizer for a special set of points in the underlying projective geometry of $PGL(2,q)$.

\subsection{Discussion}
\label{subsec:efc-Ingleton}


To measure ``how much'' the Ingleton inequality is violated, or how effective a set of subgroups is in terms of violating Ingleton, we need to compare the difference of the two sides of \eqref{equation:Ingletonh} for the corresponding entropy vector, i.e.
\[ \Delta_{h} \triangleq h_1+h_2+h_{34}+h_{123}+h_{124} - (h_{12}+h_{13}+h_{14}+h_{23}+h_{24}). \]
Translating to the finite group context it equals $\log\frac{RHS}{LHS}$ of \eqref{equation:Ingletong}. Thus we can make the following definition to measure the extent to which Ingleton is violated.
\begin{definition}
For a 4-tuple of subgroups $\tau = (G_i : 1\leq i\leq 4)$, we define the \emph{Ingleton ratio} to be
\begin{equation}\label{equation:Ingletonr}
r(\tau)=\frac{|G_{12}||G_{13}||G_{14}||G_{23}||G_{24}|}{|G_1||G_2||G_{34}||G_{123}||G_{124}|}.
\end{equation}
\end{definition}

Clearly $\Delta_{h} = \log r$ and Ingleton is violated iff $r>1$. The family $PGL(2,q)$ have the Ingleton ratio
\[ r = \frac{4(q-1)}{3q}, \]
which approaches $4/3$ when $q$ is large.

However, the Ingleton ratio is not precise enough to characterize the effectiveness of an Ingleton violation instance. Observe that $\overline{\Gamma_n^*}$ is a cone, and in fact, as remarked in\cite{Boston-Nan-NewVioIngelton}, adding an entropy vector to itself yields another entropy vector. Thus we can arbitrarily increase the Ingleton ratio by joining copies of a violation instance. For example, if $\tau = (G_i : 1\leq i\leq 4)$ is such an instance, for each integer $N$ let $G' = \times_{k=1}^{N}G \triangleq G\times\cdots\times G$ be the direct product of $N$ copies of $G$ and define $\tau' = (G_i' : 1\leq i\leq 4)$ with $G_{i}' = \times_{k=1}^{N}G_{i}$ for each $i$. Then the Ingleton ratio $r(\tau') = [r(\tau)]^{N}$, which grows unbounded when $N\to\infty$.

Therefore we need to consider the scaled version of $\Delta_{h}$, to be able to measure the effectiveness of an Ingleton violation. In \cite{DFZ-NonShannon4} Dougherty \textit{et al.} use the full joint entropy $h_{1234}$ as a scaling factor to avoid the problem above:
\begin{definition}
For an entropy vector $h = (h_\alpha:\emptyset \neq \alpha \subseteq \{1,2,3,4\})$, define the \emph{Ingleton score} to be
\[ \sigma(h) = -\frac{\Delta_{h}}{h_{1234}}. \]
In the context of groups, the Ingleton score of a 4-tuple $\tau$ of subgroups of $G$ is
\[ \sigma(\tau) = \frac{-\log r(\tau)}{\log(|G|/|G_{1234}|)}. \]
\end{definition}

Note that Ingleton fails iff $\sigma<0$, and a lower score means a larger violation. Essentially this definition forms a ray starting from the origin and passing through the point in $\mathbb{R}^{2^4-1}$ corresponding to an entropy vector, then finds its intersection with the hyperplane $h_{1234}=1$ and computes $-\Delta_{h}$ for that point to measure the Ingleton violation. The best Ingleton score in the family $PGL(2,q)$ is attained when $q = 13$, with $\sigma = -0.0270$. In \cite{Boston-Nan-NewVioIngelton} many violations obtained have lower Ingleton scores, hence are more effective than $PGL(2,q)$. In \cite{DFZ-NonShannon4} a conjecture concerning the lowest Ingleton score attainable by an arbitrary entropy vector is proposed, but has been refuted recently by Mat\'u\v s and Csirmaz\cite{Matus-Csirmaz-Rej4atom}.

A perhaps more geometrically meaningful scaling factor is the 2-norm of the entropy vector, as proposed in \cite{Shadbakht-Hassibi-MCMC}:
\begin{definition}
For an entropy vector $h = (h_\alpha:\emptyset \neq \alpha \subseteq \{1,2,3,4\})$, define the \emph{Ingleton violation index} to be
\[ \iota(h) = \frac{\Delta_{h}}{\|h\|_{2}} = \frac{\Delta_{h}}{\sqrt{h^{T}h}}. \]
\end{definition}

Essentially this definition measures the ``sine'' of the angle between an entropy vector and the Ingleton hyperplane $\Delta_{h} = 0$. The Ingleton inequality fails iff $\iota>0$, and a larger index means a larger violation. Note that two entropy vectors might have the same violation index but different Ingleton scores, and vice versa. The best Ingleton violation index in the family $PGL(2,q)$ is again attained when $q = 13$, with $\iota = 0.0082$, whereas for an arbitrary entropy vector the best $\iota$ found in literature is $0.0276$ using quasi-uniform distributions\cite{Shadbakht-thesis}.\\


Next we discuss two directions for generalizing the above Ingleton-violating family and finding new violations. On the one hand, $PGL(2,q)$ is the quotient group of $GL(2,q)$, so supposedly $GL(2,q)$ should have a richer choice of subgroups violating Ingleton inequality. This approach is explored in the next section. On the other hand, since the subgroups in the $PGL(2,q)$ family have simple but fundamental interpretations in terms of group actions, we can generalize them in this framework. In particular, we obtain two new families of violations in $PGL(n,q)$ for general $n$, and further generalize to an abstract construction using 2-transitive groups. Since this approach is more abstract and requires more background knowledge, we defer it to Section~\ref{section:Asch}.

\section{Ingleton Violations in $GL(2,q)$}
\label{section:gl}

As $PGL(2,q)$ is the quotient group of $GL(2,q)$ modulo the subgroup $V_q$ of scalar
matrices, naturally one may ask if the general linear groups also violate
Ingleton. In fact, the following lemma shows that there is at least one set of subgroups in $GL(2,q)$ that violates Ingleton for all finite field orders $q \ge 5$:

\begin{lemma}\label{lemma:preimg}
If $G$ is a finite group with a normal subgroup $N$ such that $H \triangleq G/N$ has a set of Ingleton-violating subgroups, then the preimages of these subgroups under the natural homomorphism $g \mapsto gN$ are subgroups of $G$ that also violate Ingleton.
\end{lemma}
\begin{proof}
Let $(H_i: 1\leq i\leq 4)$ be a set of Ingleton-violating subgroups in $H$. Define $G_i$ to be the preimage of $H_i$ under the natural homomorphism, then $G_i$ is a group containing $N$ for each $i$. By the Lattice Isomorphism Theorem (see e.g. \cite{absAlg}), for any nonempty subset $\alpha \subseteq \{1, 2, 3, 4\}$, $G_\alpha / N = H_\alpha$, and so $|G_\alpha | = |H_\alpha| \cdot |N|$. Thus by checking the orders in \eqref{equation:Ingletong}, $(G_i: 1\leq i\leq 4)$ also violate Ingleton.
\end{proof}

Searching with GAP, we find $GL(2,5)$ to be the smallest
general linear group that violates Ingleton. Up to subscript symmetries and
conjugations, it has 15 sets of Ingleton-violating subgroups. We would like to analyze their structures and generalize them for $q\geq 5$ if possible.

Throughout this section, we always assume $q$ is a finite field order, and $p$ is the characteristic of $\f_q$. We begin our analysis by identifying the preimages of the Ingleton-violating subgroups in the previous section under the natural homomorphism
\[ \pi:GL(2,q)\to GL(2,q)/V_q = PGL(2,q), \]
according to Lemma~\ref{lemma:preimg}. With no surprise, when $q=5$ these correspond to one of the 15 violation instances in $GL(2,5)$, and they take on nice matrix structures similar to the subgroups in Section~\ref{section:pgl}. Based on this set of subgroups we have 10 other instances, all of which are essentially its variants: each instance differs from the preimages at exactly one subgroup (either $G_1$ or $G_2$). These 11 violation instances can be easily extended to families of Ingleton-violating subgroups in $GL(2,q)$ for $q\geq 5$, sometimes with an extra condition. The remaining 4 instances cannot be derived directly from the preimages; however, they are interrelated and all their subgroups are equal or conjugate to some known subgroups from the previous instances. They also generalize to Ingleton-violating families in $GL(2,q)$ with some extra conditions.

Table~\ref{table:conditionspq} summarizes how the generalization of these instances depends on the values of $p$ and $q$. We can see that when $p=2$, these 15 instances collapse to only 6 dinstinct ones; also some instances need specific conditions on $p$ and $q$ to violate Ingelton.

%
\begin{table}[!t]
\caption{\hspace{0.1in}(a) Identical Instances when $p=2$\hspace{0.4in} (b) Cases when Ingleton is not Violated}
\label{table:conditionspq}
\hspace{0.045\textwidth}\begin{minipage}[t]{0.5\textwidth}\centering
\begin{tabular}{|c|c|}
  \hline
  Instance & Identical\vspace*{-1.5ex} \\
  No. & Instance(s) \\
  \hline
  1 & 5 \\
  \hline
  2 & 3, 4 \\
  \hline
  6 & 8, 10 \\
  \hline
  7 & 9, 11 \\
  \hline
  12 & 13 \\
  \hline
  14 & 15 \\
  \hline
\end{tabular}
\end{minipage}
\hspace{-0.1\textwidth}\begin{minipage}[t]{0.5\textwidth}\centering
\begin{tabular}{|c|c|c|}
	\hline
    & \phantom{$\fq$ odd} & $p\neq2$, \\
   Instance No. & $p=3$ & $\fq$ odd \\
  \hline
   8, 9 & & $\times$ \\
  \hline
   12, 14& $\times$ & \\
  \hline
   13, 15& $\times$ & $\times$ \\
  \hline
\end{tabular}
\end{minipage}
\end{table}

\begin{table}[!t]
\caption{Orders of Subgroups and Intersections}
\label{table:order}
\centering
\begin{tabular}{c|c|c|c|c|c|c|c|c|c}
\hline
Ins. No.&
$|G_1|$     &$|G_2|$        &$|G_3|$        &$|G_{34}|$                 &$|G_{123}|$    &
$|G_{12}|$  &$|G_{13}|$     &$|G_{23}|$     &$LHS - RHS$ in \eqref{equation:Ingletong}\\
\hline
0&
6           &$q(q-1)$       &$2(q-1)$       &1                          &1              &
2           &2              &$q-1$          &$2(q-1)(4-q)$\\
\hline
1&
$6(q-1)$    &$q(q-1)^2$     &$2(q-1)^2$     &$q-1$                      &$q-1$          &
$2(q-1)$    &$2(q-1)$       &$(q-1)^2$      &$2(q-1)^6(4-q)$\\
\hline
2,4&
6           &$q(q-1)^2$     &$2(q-1)^2$     &$q-1$                      &1              &
2           &2              &$(q-1)^2$      &$2(q-1)^3(4-q)$\\
\hline
3&
12          &$q(q-1)^2$     &$2(q-1)^2$     &$q-1$                      &2              &
4           &4              &$(q-1)^2$      &$16(q-1)^3(4-q)$\\
\hline
5&
$3(q-1)$    &$q(q-1)^2$     &$2(q-1)^2$     &$q-1$                      &$\fq$          &
$q-1$       &$q-1$          &$(q-1)^2$      &$\frac{1}{4}(q-1)^6(4-q)$\\
\hline
6--9&
$6(q-1)$    &$q(q-1)$       &$2(q-1)^2$     &$q-1$                      &1              &
2           &$2(q-1)$       &$q-1$          &$2(q-1)^3(4-q)$\\
\hline
10,11&
$6(q-1)$    &$2q(q-1)$      &$2(q-1)^2$     &$q-1$                      &2              &
4           &$2(q-1)$       &$2(q-1)$       &$16(q-1)^3(4-q)$\\
\hline
12--15&
6           &$q(q-1)$       &$q(q-1)$       &1                          &1              &
2           &2              &$q-1$          &$2(q-1)(4-q)$\\
\hline
8',9'&
$6(q-1)$    &$q(q-1)$       &$2(q-1)^2$     &$q-1$                      &2              &
2           &$2(q-1)$       &$q-1$          &$8(q-1)^3(2q+1)$\\
\hline
13',15'&
6           &$q(q-1)$       &$q(q-1)$       &2                          &1              &
1           &1              &$q-1$          &$(q-1)(11q+1)$\\
\hline
\end{tabular}
\end{table}

In Table~\ref{table:order}, the orders of the subgroups for the cases we have explored in $PGL(2,q)$ and $GL(2,q)$ are listed. No.~0 denotes the instance in $PGL(2,q)$, and No.~1--15 denote the generalizations of the 15 violation instances in $GL(2,5)$ to $GL(2,q)$. Since all instances have the subgroup order symmetries
\[ |G_3|=|G_4|,\quad |G_{123}|=|G_{124}|,\quad |G_{13}|=|G_{14}|,\quad |G_{23}|=|G_{24}|, \]
only one of each pair of orders is listed. Note that when $p=2$, there are only 6 such dinstinct generalizations, which are Instances~1, 2, 6, 7, 12 and 14. Thus for the order calculation of all other instances in $GL(2,q)$ assume $p\neq2$. Moreover, No.~8', 9', 13' and 15' correspond to Instances~8, 9, 13 and 15 when $p\neq2$ but $\fq$ is odd, in which case Ingleton is satisfied. Finally, the order calculation for Instances~12--15 only works for $p\neq3$. From Table~\ref{table:order}, we can calculate that all violation instances in the table have the same Ingleton ratio
$r = 4(q-1)/(3q)$, which is the same as the family $PGL(2,q)$. But the scaling factors for both the Ingleton score and the violation index are no larger than $PGL(2,q)$ in these instances, so they are no more effective.

In the following, we present all of these extended violation families, with Section~\ref{subsec:vio1} being the set of preimage subgroups, Sections~\ref{subsec:vio2-5} and \ref{subsec:vio6-11} its 10 variants, and Section~\ref{subsec:vio12-15} the remaining 4 instances. We continue to use the notations from Section~\ref{section:pgl} with $t$ being a primitive element of $\f_q$, but we redefine
\[ N = \left\{\left.A_\alpha \right| \alpha\in\f_q\right\} =
\langle A_{\xi_1},A_{\xi_2},\ldots,A_{\xi_m} \rangle \cong \langle A_{\xi_1}\rangle\times\langle A_{\xi_2} \rangle\times\ldots\times\langle A_{\xi_m}\rangle \cong \z_p^m. \]
In addition, we introduce the following matrices and subgroups in $GL(2,q)$ to facilitate our presentation:
\[ B' = \begin{bmatrix} -1 & 0 \\ 0 & t \end{bmatrix},\qquad
P = \begin{bmatrix} t & 0 \\ 0 & 1 \end{bmatrix},\qquad
P' = \begin{bmatrix} t & 0 \\ 0 & -1 \end{bmatrix}, \]
\[ M = \langle C, B_1\rangle = \left\{I,\ \begin{bmatrix} 0 & 1 \\ -1 & -1 \end{bmatrix},\
\begin{bmatrix} -1 & -1 \\ 1 & 0 \end{bmatrix},\ \begin{bmatrix} 1 & 0 \\ -1 & -1 \end{bmatrix},\ \begin{bmatrix} 0 & 1 \\ 1 & 0\end{bmatrix},\ \begin{bmatrix} -1 & -1 \\0 & 1 \end{bmatrix} \right\}, \]
\[ K = \langle N, B \rangle =\left\{\left.\begin{bmatrix} 1 & 0 \\ \alpha & \beta \end{bmatrix}\right|
\begin{array}{l} \alpha \in \f_q,\\ \beta \in \f_q^\times \end{array}\right\},\quad
K' = \langle N, B' \rangle =\left\{\left.\begin{bmatrix} (-1)^k & 0 \\ \alpha & t^k \end{bmatrix}\right|
\begin{array}{l} \alpha \in \f_q,\\ k \in \K_q \end{array}\right\}, \]
\[ J = \langle N, P \rangle =\left\{\left.\begin{bmatrix} \beta & 0 \\ \alpha & 1 \end{bmatrix}\right|
\begin{array}{l} \alpha \in \f_q,\\ \beta \in \f_q^\times \end{array}\right\},\quad
J' = \langle N, P' \rangle =\left\{\left.\begin{bmatrix} t^k & 0 \\ \alpha & (-1)^k \end{bmatrix}\right|
\begin{array}{l} \alpha \in \f_q,\\ k \in \K_q \end{array}\right\}. \]
Note that when $p=2$, we have $-1=1$, so $B'=B$, $P'=P$, and $K'=K$, $J'=J$. Also note that $M$ and $K$ precisely correspond to $G_1$ and $G_2$ in Section~\ref{section:pgl}, respectively. The group $M$ is isomorphic to $D_6 \cong S_3$, while the other four groups are all semidirect products $\z_p^m\rtimes \z_{q-1}$, with $K\cong J$ and $K'\cong J'$. Moreover, $K$ and $J$ have generalized flower structures for all $q>2$. However, if $p\neq2$, $K'$ and $J'$ only have flower structures when $\fq$ is even, in which case they are also isomorphic to $K$. (See Section~\ref{subsec:app_subgrp} in Appendices for proofs.) This turns out to be a necessary condition to violate Ingleton in all the instances where $K'$ and $J'$ are involved.

\subsection{Instance 1: The Preimage Subgroups}
\label{subsec:vio1}

To obtain the preimage $H_0$ of a subgroup $H\leq PGL(2,q)$ under $\pi$, we can generate $H_0$ in $GL(2,q)$ with the generators of $H$ (without overlines) and $tI$, since $V_q = \langle tI\rangle \cong \z_{q-1}$.

$G_1 = \langle tI, C, B_1 \rangle = \langle V_q, M \rangle$. Since $V_q$ is the center of $GL(2,q)$ and intersects $M$ trivially, $G_1$ is a direct product: $G_1 = \{t^kX\,|\,X\in M, k \in \K_q\} \cong V_q\times M \cong \z_{q-1}\times S_3$.

$G_2 = \langle tI, A_{\xi_1},A_{\xi_2},\ldots,A_{\xi_m}, B \rangle = \langle tI, N, B \rangle = \langle V_q, K \rangle$. $G_2$ is the subgroups of all lower triangular matrices in $GL(2,q)$, and as $V_q\bigcap K = 1$, we have $G_2 \cong V_q\times K \cong \z_{q-1}\times (\z_p^m\rtimes \z_{q-1})$.

$G_3 = \langle tI, B_1C^2, B_3 \rangle = \langle tI, CB_1, B_3 \rangle = \langle CB_1, T \rangle$, where $T = \langle tI, B_3\rangle$. As $V_q\bigcap\langle B_3\rangle=1$, we have $T = {\{t^kB_3^m\,|\,k,m \in \K_q\}} \cong V_q\times \langle B_3\rangle \cong \z_{q-1}\times\z_{q-1}$. It is easy to check that $(t^kB_3^m)^{CB_1} = t^{k+m}B_3^{-m}\in T$, so $G_3 = \langle CB_1 \rangle\cdot T$ and $T\trianglelefteq G_3$. Furthermore, $|CB_1|=2$ and $T\bigcap\langle CB_1\rangle=1$, thus $G_3\cong T \rtimes \langle CB_1\rangle \cong (\z_{q-1}\times\z_{q-1})\rtimes\z_2$ and
\[ G_3 = \left\{\left. t^k\begin{bmatrix} 1 & 0 \\ t^m-1 & t^m \end{bmatrix},\quad
t^{k+m}\begin{bmatrix} -1 & -1 \\ 1-t^{-m} & 1 \end{bmatrix}\right|k,m \in \K_q\right\}.\]

$G_4 = \langle tI, B_1C, B_4 \rangle = \langle tI, B_1C, B \rangle = \langle B_1C, D \rangle$, where $D = \langle tI, B\rangle$. Since $V_q\bigcap\langle B\rangle=1$, we have $D = {\{t^kB^m\,|\,k,m \in \K_q\}} = $ \{all diagonal matrices in $GL(2,q)$\} $\cong V_q\times \langle B\rangle \cong \z_{q-1}\times\z_{q-1}$. Note that \[ \begin{bmatrix} \alpha & 0 \\ 0 & \beta \end{bmatrix}^{B_1C} =
\begin{bmatrix} \beta & 0 \\ 0 & \alpha \end{bmatrix} \in D, \]
so $G_4 = \langle B_1C \rangle\cdot D$ and $D\trianglelefteq G_4$. Since $|B_1C|=2$ and $D\bigcap\langle B_1C\rangle=1$, $G_4 \cong D \rtimes \langle B_1C\rangle \cong {(\z_{q-1}\times\z_{q-1})\rtimes\z_2}$. Actually $G_4$ is the subgroups of all diagonal and anti-diagonal matrices in $GL(2,q)$:
\[ G_4 = \left\{\left.\begin{bmatrix} \alpha & 0 \\ 0 & \beta \end{bmatrix}, \begin{bmatrix} 0 & \beta \\ \alpha & 0 \end{bmatrix}\right|\alpha,\beta\in\f_q^\times\right\}.\]

Calculating the intersections, we have
$G_{12} = \langle tI, B_1 \rangle \cong V_q\times \langle B_1\rangle$,
$\ G_{13} = \langle tI, CB_1 \rangle \cong V_q\times \langle CB_1\rangle$ and
$G_{14} = \langle tI, B_1C \rangle \cong V_q\times \langle B_1C\rangle$, all of which are isomorphic to $\z_{q-1}\times\z_2$. Also, $G_{23} = T,\ G_{24} = D$ and $G_{34} = G_{123} = G_{124} = \langle tI \rangle = V_q$.

From the calculation in Table~\ref{table:order}, Ingleton is violated when $q\geq 5$.

\subsection{Instances 2--5: Variants with Different $G_1$'s}
\label{subsec:vio2-5}

In all the instances in this section, only $G_1$ is different from Instance~1; it is now a \textit{proper} subgroup of $\langle tI, C, B_1 \rangle$ (see Table~\ref{table:2-5}, where the generator-form for these groups is used to better demonstrate the subgroup relations). When $p\neq2$, these instances are all distinct; however, when $p=2$, clearly Instances~3 and 4 collapse to Instance~2, while Instance~5 becomes Instance~1. From Table~\ref{table:order}, we can see that they all violate Ingleton when $q\geq 5$.

\begin{table}[!t]
\caption{$G_1$ for Instances 2--5}
\label{table:2-5}
\centering
\begin{tabular}{c|c|c|c|c}
  \hline
  Ins. No. & 2 & 3 & 4 & 5 \\
  \hline
  $G_1$ & $\langle C, B_1\rangle$ & $\langle -C, B_1 \rangle$ & $\langle C, -B_1\rangle$ & $\langle C, tB_1 \rangle$ \\
  \hline
\end{tabular}
\end{table}

\subsubsection{Instance 2}

$G_1 = M$.

$G_{12} = \langle B_1 \rangle$, $\ G_{13} = \langle CB_1 \rangle$ and $G_{14} = \langle B_1C \rangle$ are  all isomorphic to $\z_2$, and $G_{123} = G_{124} = 1$.

\subsubsection{Instance 3}

$G_1 = \langle -C, B_1 \rangle$.

We only consider the case $p\neq2$, since otherwise this is the same as Instance~2. As $|C|=3$, we have $(-C)^3=-I$ and $(-C)^4=C$. Thus $G_1 = \langle -I, C, B_1 \rangle = \langle -I, M \rangle \cong \langle -I \rangle \times M \cong \z_2\times S_3\cong D_{12}$, since $\langle -I \rangle$ is a subgroup of $V_q$ and intersects $M$ trivially. So $G_1 = \{\pm X\,|\,X\in M\}$.

Now $G_{12} = \langle -I, B_1 \rangle \cong \langle -I \rangle \times \langle B_1 \rangle$,
$\ G_{13} = \langle -I, CB_1 \rangle \cong \langle -I \rangle \times \langle CB_1 \rangle$ and
$G_{14} = \langle -I, B_1C \rangle \cong \langle -I \rangle \times \langle B_1C \rangle$,
all of which are isomorphic to $\z_2\times\z_2$. Furthermore, $G_{123} = G_{124} = \langle -I \rangle \cong \z_2$.

\subsubsection{Instance 4}

$G_1 = \langle C, -B_1\rangle$.

Here we also need only consider the case $p\neq2$. Observe that $|C|=3$, $\ \ordm{B_1}=2$ and $\left(C\cdot(-B_1)\right)^2=(CB_1)^2=I$. This gives us $G_1 = \left\{I,C,C^2,-B_1,-B_1C,-CB_1\right\}$, so $G_1 \cong D_6 \cong S_3$.

For the intersections, we have $G_{12} = \langle -B_1 \rangle$, $\ G_{13} = \langle -CB_1 \rangle$ and
$G_{14} = \langle -B_1C \rangle$ all isomorphic to $\z_2$, and $G_{123} = G_{124} = 1$.

\subsubsection{Instance 5}

$G_1 = \langle C, tB_1 \rangle$.

When $p=2$, $q$ is even. Since $|B_1|=2$ and $|t|=q-1$, we have $(tB_1)^q = tI$ and $(tB_1)^{q-1} = B_1$. Thus $G_1 = \langle tI,C,B_1 \rangle$ and this instance is the same as Instance~1.

Now assume $p\neq2$. As $q$ is odd, $|tB_1|=q-1$. When $k$ is even, $(tB_1)^k = t^kI$ and so $C^{(tB_1)^k} = C$. Otherwise $(tB_1)^k = t^kB_1$, then $C^{(tB_1)^k} = B_1CB_1 = C^{-1}$ since $(CB_1)^2=I$. So $G_1 = \langle tB_1 \rangle\cdot\langle C \rangle$ and $\langle C \rangle\trianglelefteq G_1$. Furthermore, $\langle tB_1 \rangle\bigcap\langle C \rangle=1$ and $|C|=3$, thus $G_1 \cong \langle C \rangle \rtimes \langle tB_1 \rangle \cong \z_3\rtimes\z_{q-1}$ and
$G_1 = \left\{t^kI,\, t^kC,\, t^kC^2\ \big|\ k \text{ even, } k\in\K_q\right\}
\bigcup\left\{t^kB_1,\, t^kB_1C,\, t^kCB_1\ \big|\ k \text{ odd, } k\in\K_q\right\}$.

The intersections are: $G_{12} = \langle tB_1 \rangle$, $\ G_{13} = \langle tCB_1 \rangle$ and $G_{14} = \langle tB_1C \rangle$ are all isomorphic to $\z_{q-1}$, and $G_{123} = G_{124} = \langle t^2I \rangle \cong \z_\fq$.

\subsection{Instances 6--11: Variants with Different $G_2$'s}
\label{subsec:vio6-11}

In all the instances in this section, only $G_2$ is different from Instance~1; it is now a \textit{proper} subgroup of $\langle tI, N, B \rangle$ (see Table~\ref{table:6-11}). It is easy to see that these instances are distinct when $p\neq2$; otherwise Instances~8 and 10 collapse to Instance~6, while Instances~9 and 11 become Instance~7. Thus in the analysis of Instances~8--11, we assume $p\neq2$. From Table~\ref{table:order}, Instances 6, 7, 10, 11 violate Ingleton whenever $q\geq 5$; however, if $p\neq2$, Instances 8 and 9 only violate Ingleton when in addition $\fq$ is even. Please refer to Section~\ref{subsec:app_89int} in Appendices for the calculation of subgroup intersections in Instances 8 and 9.

\begin{table}[!t]
\caption{$G_2$ for Instances 6--11}
\label{table:6-11}
\centering
\begin{tabular}{c|c|c|c|c|c|c}
  \hline
  Ins. No. & 6 & 7 & 8 & 9 & 10 & 11 \\
  \hline
  $G_2$ & $\langle N, B \rangle$ & $\langle N, P \rangle$ & $\langle N, B' \rangle$ & $\langle N, P' \rangle$ & $\langle -I, N, B \rangle$ & $\langle -I, N, P \rangle$ \\
  \hline
\end{tabular}
\end{table}

\subsubsection{Instance 6}

$G_2 = K$.

In this case, $G_{12} = \langle B_1 \rangle \cong \z_2$ and $G_{123} = G_{124} = 1$. Also $G_{23} = \langle B_3 \rangle$ and
$G_{24} = \langle B \rangle$, both of which are isomorphic to $\z_{q-1}$.

\subsubsection{Instance 7}

$G_2 = J$.

Here $G_{12} = \langle -B_1 \rangle \cong \z_2$, $G_{123} = G_{124} = 1$. Also, $G_{23} = \langle t^{-1}B_3 \rangle$ and $G_{24} = \langle  P \rangle$, both isomorphic to $\z_{q-1}$.

\subsubsection{Instance 8}

$G_2 = K'$.
\[ G_{12} = \left\{\begin{array}{cl}\langle B_1 \rangle \cong \z_2 & \text{if }\fq\text{ is even}\\
    \langle -I \rangle \cong \z_2 & \text{otherwise}\end{array}\right.,\quad
G_{123} = G_{124} = \left\{\begin{array}{cl}1 & \text{if }\fq\text{ is even}\\
    \langle -I \rangle \cong \z_2 & \text{otherwise}\end{array}\right.. \]
In this case, $G_{23} = \langle -B_3^\mfq \rangle$ and $G_{24} = \langle B' \rangle$ are both isomorphic to $\z_{q-1}$.

\subsubsection{Instance 9}

$G_2 = J'$.
\[ G_{12} = \left\{\begin{array}{cl}\langle -B_1 \rangle \cong \z_2 & \text{if }\fq\text{ is even}\\
    \langle -I \rangle \cong \z_2 & \text{otherwise}\end{array}\right.,\quad
G_{123} = G_{124} = \left\{\begin{array}{cl}1 & \text{if }\fq\text{ is even}\\
    \langle -I \rangle \cong \z_2 & \text{otherwise}\end{array}\right.. \]
Here $G_{23} = \langle tB_3^{\frac{q-3}{2}} \rangle$ and $ G_{24} = \langle P' \rangle$ are isomorphic to $\z_{q-1}$.

\subsubsection{Instance 10}

$G_2 = \langle -I, N, B \rangle$.

Now we have $G_2 = \langle -I, K \rangle \cong \langle -I \rangle \times K \cong \z_2\times (\z_p^m\rtimes \z_{q-1})$, since $\langle -I \rangle \bigcap K = 1$. Thus $G_2 = \{\pm X\,|\,X\in K\}$.

For the intersections, we have $G_{12} = \langle -I, B_1 \rangle \cong \z_2\times\z_2$ and
$G_{123} = G_{124} = \langle -I \rangle \cong \z_2$. Also,
$G_{23} = \langle -I, B_3 \rangle \cong \langle -I \rangle \times \langle B_3 \rangle$ and
$ G_{24} = \langle -I, B \rangle \cong \langle -I \rangle \times \langle B \rangle$, both isomorphic to $\z_2\times\z_{q-1}$.

\subsubsection{Instance 11}

$G_2 = \langle -I, N, P \rangle$.

Here $G_2 = \langle -I, J \rangle \cong \langle -I \rangle \times J \cong \z_2\times (\z_p^m\rtimes \z_{q-1})$, since $\langle -I \rangle \bigcap J = 1$. Thus $G_2 = \{\pm X\,|\,X\in J\}$.

Moreover, $G_{12} = \langle -I, -B_1 \rangle = \langle -I, B_1 \rangle \cong  \z_2\times\z_2$ and $G_{123} = G_{124} = \langle -I \rangle \cong \z_2$. Also, $G_{23} = \langle -I, t^{-1}B_3 \rangle \cong \langle -I \rangle \times \langle t^{-1}B_3 \rangle$ and $G_{24} = \langle -I, P \rangle \cong \langle -I \rangle \times \langle P \rangle$ are both isomorphic to $\z_2\times\z_{q-1}$.

\subsection{Instances 12--15}
\label{subsec:vio12-15}

\begin{table}[!t]
\caption{Subgroups for Instances 12--15}
\label{table:subgroups}
\centering
\begin{tabular}{c|c|c|c|c}
\hline
Ins. No. & $G_1$ & $G_2$ & $G_3$ & $G_4$\\
\hline
12 & $M$ & $\langle N,B \rangle^{\protect\phantom{E}}$ & $\langle N,P \rangle^E$ & $\langle N,P \rangle^Q$\\
\hline
13 & $M$ & $\langle N,B' \rangle^{\protect\phantom{E}}$ & $\langle N,P' \rangle^E$ & $\langle N,P' \rangle^Q$\\
\hline
14 & $M$ & $\langle N,P \rangle^E$ & $\langle N,B \rangle^{\protect\phantom{E}}$ & $\langle N,B \rangle^W$\\
\hline
15 & $M$ & $\langle N,P' \rangle^E$ & $\langle N,B' \rangle^{\protect\phantom{E}}$ & $\langle N,B' \rangle^W$\\
\hline
\end{tabular}
\end{table}

For these last four instances, $G_1$ is always $M$, $G_2$--$G_4$ are equal or conjugate to one of $K,K',J,J'$, as listed in Table~\ref{table:subgroups}. Thus $G_2$--$G_4$ are all semidirect products $\z_p^m\rtimes \z_{q-1}$ and the structures of $G_3$ and $G_4$ are different from all previous instances. The conjugators $E,Q,W$ and the elements of new subgroups are listed as follows.
\[ E = \begin{bmatrix} -1 & 1 \\ 1 & 0 \end{bmatrix},\qquad
Q = \begin{bmatrix} 2 & 1 \\ 1 & 0 \end{bmatrix},\qquad
W = \begin{bmatrix} 0 & 1 \\ -1 & 1 \end{bmatrix}. \]
\[ J^E = \langle N,P \rangle^E = 
\left\{\left.\begin{bmatrix} 1-v & v \\ 1-u-v & u+v \end{bmatrix}\right|
\begin{array}{l} u\in\f_q^\times, \\ v\in\f_q \end{array}\right\}, \]
\[ (J')^E = \langle N,P' \rangle^E = \left\{\left.
\begin{bmatrix} (-1)^j-\alpha & \alpha \\ (-1)^j-t^j-\alpha & t^j+\alpha \end{bmatrix}\right|
\begin{array}{l} \alpha\in\f_q, \\ j\in\K_q \end{array}\right\}, \]
\[ J^Q = \langle N,P \rangle^Q =
\left\{\left.\begin{bmatrix} 1+2y & y \\ 2(x-2y-1) & x-2y \end{bmatrix}\right|
\begin{array}{l} x\in\f_q^\times, \\ y\in\f_q \end{array}\right\}, \]
\[ (J')^Q = \langle N,P' \rangle^Q = \left\{\left.
\begin{bmatrix} (-1)^i+2\beta & \beta \\ 2\left(t^i-2\beta-(-1)^i\right) & t^i-2\beta \end{bmatrix}\right|
\begin{array}{l} \beta\in\f_q, \\ i\in\K_q \end{array}\right\}, \]
\[ K^W = \langle N,B \rangle^W =
\left\{\left.\begin{bmatrix} x & y \\ 0 & 1 \end{bmatrix}\right|
\begin{array}{l} x\in\f_q^\times, \\ y\in\f_q \end{array}\right\} =
\left\{\left.X^T\,\right|\,X\in J\right\}, \]
\[ (K')^W = \langle N,B' \rangle^W = \left\{\left.
\begin{bmatrix} t^i & \beta \\ 0 & (-1)^i \end{bmatrix}\right|
\begin{array}{l} \beta\in\f_q, \\ i\in\K_q \end{array}\right\} =
\left\{\left.X^T\,\right|\,X\in J'\right\}. \]

As mentioned in Table~\ref{table:conditionspq}, Instances 12--15 do not violate Ingleton when $p=3$. The reasons are as follows. If $p=3$, then $2=-1$, so $E=Q$ and $M\leq J^E$.
Thus in Instance~12 we have $G_3=G_4$ and $G_1\leq G_3$, while in Instances~13 and 14 we have $G_3=G_4$ and $G_1\leq G_2$ respectively. So these three instances satisfy Conditions~\ref{condition:distinct} and/or \ref{condition:subgroup}. Instance~15, however, satisfies Condition~\ref{condition:g1g2} in this case (see Section~\ref{subsec:app_15p3} in Appendices).

Besides, we also need $p\neq2$ to make Instances~13 and 15 distinct: otherwise they collapse to Instances~12 and 14 respectively. Thus in the rest of this section, we always assume $p\neq3$, while for Instances~13 and 15 we assume $p>3$. From Table~\ref{table:order}, Instances 12 and 14 violate Ingleton when $q\geq 5$ (and of course, $p\neq3$), while if $p\neq2$, Instances 13 and 15 only violate Ingleton when in addition $\fq$ is even. Please refer to Section~\ref{subsec:app_1215int} in Appendices for the intersection calculations.

\subsubsection{Instance 12}

$G_2 = K, G_3 = J^E, G_4 = J^Q$.

We have $G_{12} = \langle B_1 \rangle$, $\ G_{13} = \langle B_1C \rangle$ and $G_{14} = \langle CB_1 \rangle$ all isomorphic to $\z_2$, and $G_{34} = G_{123} = G_{124} = 1$. Furthermore,
\[ G_{23} = \left\{\left.\begin{bmatrix} 1 & 0 \\ 1-t^j & t^j \end{bmatrix}\right| j \in \K_q \right\} = \langle P \rangle^E,\quad
G_{24} =  \left\{\left.\begin{bmatrix} 1 & 0 \\ 2(t^i-1) & t^i \end{bmatrix}\right| i \in \K_q \right\} = \langle P \rangle^Q \]
both are isomorphic to $\z_{q-1}$.

\subsubsection{Instance 13}

$G_2 = K', G_3 = (J')^E, G_4 = (J')^Q$.

When $\fq$ is even, $G_{12},G_{13},G_{14}$ and $G_{34}$ are the same as in Instance 12. Otherwise $G_{12} = G_{13} = G_{14} = 1$ and $G_{34} = \langle -I \rangle \cong \z_2$. $G_{123}$ and $G_{124}$ are always trivial. Also,
\[ G_{23} = \left\{\left.\begin{bmatrix} (-1)^j & 0 \\ (-1)^j-t^j & t^j \end{bmatrix}\right|
j \in \K_q \right\} = \langle P' \rangle^E,\quad
G_{24} =  \left\{\left.\begin{bmatrix} (-1)^i & 0 \\ 2\left(t^i-(-1)^i\right) & t^i \end{bmatrix}\right|
i \in \K_q \right\} = \langle P' \rangle^Q \]
are both isomorphic to $\z_{q-1}$.

\subsubsection{Instance 14}

$G_2 = J^E, G_3 = K, G_4 = K^W$.

Observe that $G_2$ and $G_3$ are obtained from swapping the corresponding subgroups from Instance 12. Therefore $G_{12}$ and $G_{13}$ are also swapped while $G_{23}$ remains the same. It turns out that
$G_{14},G_{34},G_{123}$ and $G_{124}$ are also the same as in Instance 12. Furthermore,
\[ G_{24} =  \left\{\left.\begin{bmatrix} t^i & 1-t^i \\ 0 & 1 \end{bmatrix}\right| i \in \K_q \right\} = \langle B \rangle^W \cong \z_{q-1}. \]

\subsubsection{Instance 15}

$G_2 = (J')^E, G_3 = K', G_4 = (K')^W$.

In this case, $G_2$ and $G_3$ from Instance 13 are swapped to yield the corresponding subgroups here. So $G_{12}$ and $G_{13}$ are also swapped while $G_{23}$ stays the same. Moreover, $G_{14},G_{34},G_{123}$ and $G_{124}$ are the same as in Instance 13, both when $\fq$ is even and otherwise. Finally,
\[ G_{24} =  \left\{\left.\begin{bmatrix} t^i & (-1)^i-t^i \\ 0 & (-1)^i \end{bmatrix}\right| i \in \K_q \right\} = \langle B' \rangle^W \cong \z_{q-1}. \]

\section{Interpretation and Generalizations of Violation in $PGL(2,q)$ using Theory of Group Actions}
\label{section:Asch}
Instead of invertible matrices, we can also regard a general linear group as the group of all invertible linear transformations on a vector space. In this section, we take this point of view and consider the actions of linear groups on their corresponding projective geometries. Such actions induce a permutation representation for each general linear group on its projective geometry, and the projective linear groups are naturally defined in this framework. Using the theory of group actions, we show that the Ingleton violation in $PGL(2,q)$ from Section~\ref{section:pgl} has a nice interpretation: each subgroup is some sort of stabilizer for a set of points in the projective geometry. Furthermore, based on this understanding, we generalize the construction in $PGL(2,q)$ to two new families of Ingleton violations in $PGL(n,q)$ for a general $n$.\footnote{Note that with Lemma~\ref{lemma:preimg}, the families in $PGL(n,q)$ can also be easily extended to families of violations in $GL(n,q)$.} Finally, we provide an abstract construction in 2-transitive groups generalizing these ideas.

Throughout this section we assume basic knowledge in the theory of group actions, which can be found in standard group theory textbooks. In particular, we make extensive use of the orbit-stabilizer theorem, which says the order of the orbit of an element is equal to the index of it stabilizer (see e.g. \cite[Sec. 4.1, Prop. 2]{absAlg}). Most notations are standard abstract algebra notations, see e.g. \cite{absAlg}; the rest are introduced when they first appear. Note that this section is more abstract than the others and assumes more background knowledge in abstract algebra.

This section is mostly based on Prof. M.~Aschbacher's correspondences with us. We have furnished various details and explanations for clarity.

\subsection{Preliminaries for Linear Groups}

Let $V$ be an $n$-dimensional vector space over a field $F$. Recall $GL(V)$ and $SL(V)$ are the general linear group and special linear group on $V$, respectively. They are examples of groups of Lie type, a notion which is not totally well defined.

Each group $G$ of Lie type possesses a \emph{building}, a simplicial complex on which $G$ is represented as a group of automorphisms. A (abstract) \emph{simplicial complex} consists of a set $X$ of \emph{vertices} together with a collection of nonempty subsets of $X$ called \emph{simplices}; the only axiom says that each nonempty subset of a simplex is a simplex.

\begin{example}
Let $X$ be a partially ordered set. The \emph{order complex} of $X$ is the simplicial complex with vertex set $X$ and with the simplices the nonempty chains in the poset.
\end{example}

\begin{example}
The \emph{projective geometry} $PG(V)$ of $V$ is the poset of nonzero proper subspaces of $V$, partially ordered by inclusion. The building of $GL(V)$ and $SL(V)$ is the order complex of this poset. Of course $GL(V)$ permutes the subspaces of $V$, supplying a representation of $GL(V)$ on $PG(V)$ whose kernel is the subgroup of scalar maps. The images of $GL(V)$ and $SL(V)$ in the automorphism group $Aut(PG(V))$ are the \emph{projective general linear group} $PGL(V)$ and \emph{projective special linear group} $PSL(V)$. Write $GL(n,F)$, $SL(n,F)$, $PGL(n,F)$, $PSL(n,F)$ for the corresponding group when $\mathrm{dim}(V) = n$ and the field is $F$.
\end{example}

\begin{example}
Specialize to the case $n = 2$. Then $PG(V)$ consists of the \emph{points} of $V$; i.e. the 1-dimensional subspaces of $V$. This is the so-called \emph{projective line}. Let $\X = \{x_{1},x_{2}\}$ be a basis of $V$. We regard the projective line as $\Omega = F \cup \{\infty\}$, where $\infty$ denotes $Fx_{1}$ and for $e \in F$, $e$ denotes $F(ex_{1} + x_{2})$. Then given an invertible matrix
\[ M(a,b,c,d) = \begin{bmatrix} a & b \\ c & d \end{bmatrix} \]
in $GL(V)$, one can check that, subject to the identification of $PG(V)$ with $\Omega$, $M(a,b,c,d)$ acts on $\Omega$ via
\[ M(a,b,c,d): x \mapsto \frac{ax+b}{cx+d}, \]
where arithmetic involving $\infty$ is suitably interpreted; e.g. $(a\infty + b)/(c\infty + d) = a/c$ if $c \neq 0$ and $\infty$ if $c = 0$. So we can regard $PGL(V ) = PGL(2,F)$ as the group of these projective linear maps $M(a, b, c, d)$, $ad - bc \neq 0$ on the projective line $\Omega$.

The following result is well known and easy to prove:

\begin{lemma}\label{lem:sh3trans}
$PGL(2,F)$ is \emph{sharply 3-transitive} on the projective line $PG(V)$. That is, $PGL(V)$ is transitive on ordered 3-tuples of distinct points, and only the identity fixes three points.
\end{lemma}

Next we introduce several types of subgroups for these linear groups.

A \emph{Borel subgroup} of a group $G$ of Lie type is the stabilizer of a maximal simplex in its building.

\begin{example}
A maximal simplex in $PG(V)$ is a flag $\tau = ( 0<V_{1} <\cdots<V_{n-1} <V )$, where $\mathrm{dim}(V_{k}) = k$. If we pick a basis $\X = \{x_{1},...,x_{n}\}$ for $V$ such that $V_{k} = \langle x_{i} : 1 \leq i \leq k \rangle$, then the Borel subgroup stabilizing $\tau$ is the subgroup whose matrices with respect to $\X$ are the upper triangular invertible matrices.
\end{example}

Let $G = PGL(2,F)$. By definition, the stabilizers $G_{Fx_{1}} = G_{\infty}$ and $G_{Fx_{2}} = G_{0}$ are both Borel subgroups of $G$. The matrices of these subgroups are upper triangular and lower triangular respectively. As $G$ is transitive on $\Omega$, for each of $u=\infty,0$ we have the bijection $gG_{u} \mapsto g(u)$ of the coset space $G/G_{u}$ with $\Omega$ (by orbit-stabilizer theorem).
\end{example}

Buildings have certain special subcomplexes called \emph{apartments}. For a group $G$ of Lie type, the pointwise stabilizer of an apartment is called a \emph{Cartan subgroup} of $G$.

\begin{example}
In the projective geometry, the apartments are of the form $\Sigma(\X)$ for $\X = \{x_{1},\cdots,x_{n}\}$ a basis for $V$, where $\Sigma(\X)$ consists of the subspaces spanned by nonempty proper subsets of $\X$. The matrices in the Cartan subgroup stabilizing $\Sigma(\X)$ are the diagonal matrices.

Suppose $n = 2$. Then $\Sigma(\X) = \{Fx_{1},Fx_{2}\} = \{\infty,0\}$ is just a pair of points. The \emph{global stabilizer} $G(u,v)$ of a pair of points is the subgroup of $G$ permuting the 2-subset $\{u, v\}$. In $G = PGL(2,F)$ it is (usually) the normalizer of the Cartan subgroup and dihedral. Furthermore, $G_{0} \cap G(0,\infty) = G_{0,\infty}$ is a Cartan subgroup isomorphic to the multiplicative group $F^{\times}$ of $F$.
\end{example}

Let $G$ be $GL(V)$ or $PGL(V)$ in the rest of this section.

An element of $GL(V)$ is {\it unipotent} if all its eigenvalues are 1. A subgroup of $GL(V)$ is {\it unipotent} if all its elements are unipotent. The {\it unipotent radical} $Q(H)$ of a subgroup $H$ of $GL(V)$ is the largest normal unipotent subgroup of $H$. For example if $F$ is finite of characteristic $p$, then $Q(H)$ is the largest normal $p$-subgroup of $H$. Passing to images in $PGL(V)$, we have the corresponding notions in that group also.

A subgroup $H$ of $G$ is a {\it parabolic} if $H$ is the stabilizer of a simplex in the projective geometry $PG(V)$. Thus for example Borel subgroups are parabolics, and indeed the parabolics are the overgroups of the Borel subgroups.

\begin{example}\label{example:parabolic}
Let $F=\mathbb{F}_q$, $U$ an $m$-dimensional subspace of $V$ with $0<m<n$, $G=GL(V)$, and $H=N_G(U)$ the (global) stabilizer of $U$ in $G$. As $\{U\}$ is a simplex in $PG(V)$, $H$ is a parabolic. Pick a complement $W$ to $U$ in $V$, and let $\X_{1}$ and $\X_{2}$ be bases for $U$ and $W$ respectively. Then the matrices of $H$ with respect to $\X_{1}\cup\X_{2}$ have the form
$\begin{bmatrix} K & L \\ 0 & R \end{bmatrix}$
with $K$ and $R$ invertible. Define
\[ q_n=q^{n(n-1)/2},\quad M_k=\prod_{i=1}^k(q^i-1) \]
for $1\leq k\leq n$, then
\[ |GL(k,q)|=q_{k}M_{k}, \]
\[ |H|=|GL(m,q)|\cdot|GL(n-m,q)|\cdot q^{m(n-m)}=q_nM_mM_{n-m}. \]
Furthermore, in $PGL(V)$ the image of $H$ has order $q_nM_mM_{n-m}/(q-1)$.
\end{example}

\subsection{Interpretation of the Ingleton Violation in $PGL(2,q)$}
Let $F = \mathbb{F}_{q}$ and $G = PGL(2,q) = PGL(2,\mathbb{F}_{q})$. In the Ingleton violation construction in Section~\ref{section:pgl} we have a 4-tuple of subgroups $\rho = (G_{i} : 1 \leq i \leq 4)$ of $G$. The group $G_{2} = G_{Fx_{2}} = G_{0}$ is a Borel subgroup. The subgroups $G_{3}$ and $G_{4}$ are isomorphic to the dihedral group $D_{2(q-1)}$ of order $2(q-1)$, and their intersection $G_{2i}$ with $G_{2}$ is cyclic of order $q-1$ and with $G_{34}$ of order 1. This forces $G_{2i}$, $i = 3,4$, to be distinct Cartan subgroups of $G_{2}$, and hence $G_{i} = G(0,e_{i})$ for some $e_{i} \in F$. In fact from the forms of the matrices in $G_{3}$ and $G_{4}$ it is easy to check that $e_{3} =-1$ and $e_{4} =\infty$.

Finally $G_{1} \cong S_{3}$ with $G_{1i}$ being the three subgroups of $G_{1}$ of order 2 for $2 \leq i \leq 4$. For $2\leq i\leq 4$ let $G_{1i} = \langle t_{i} \rangle$, and for $1\leq j\leq 4$ let $\Delta_{j}$ be the orbit of $G_{j}$ on $\Omega$ containing $0$. Then $|\Delta_{j}|=|G_{j} :G_{2j}|=n_{j}$ where $n_{3} =n_{4} =2$ and $n_{1} =3$. Indeed $\Delta_{i} =\{0,t_{i}(0)\}$ for $i=3,4$, with $\Delta_{3} =\{0,-1\}$ and $\Delta_{4} = \{0,\infty\}$. Then as $G_{1} =\langle t_{3},t_{4}\rangle$ and $n_{1} =3$, $\Delta=\Delta_{1} =\{0,-1,\infty\}$. But as $G$ is sharply 3-transitive, the global stabilizer $G(\Delta)$ is isomorphic to $S_{3}$. Hence $G_{1} = G(\Delta)$, and is determined by $G_{2}$, $G_{3}$ and $G_{4}$.

Hence the 4-tuple $\rho$ is determined by the ordered triple $(0,-1,\infty)$ with the four subgroups being various (global) stabilizers on it. Furthermore, given an arbitrary ordered triple $(\alpha,\beta,\gamma)$ of distinct points in $\Omega$, we can construct a 4-tuple $\rho'$ in the same fashion, where $G_2=G_{\alpha}$, $G_3=G(\alpha,\beta)$, $G_4=G(\alpha,\gamma)$, and $G_1=G(\alpha,\beta,\gamma)$. Since $G$ is 3-transitive on $\Omega$, by the same element in $G$ all four subgroups in $\rho'$ are conjugate to their counterparts in $\rho$. In particular, the new tuple $\rho'$ also violates Ingleton.

With respect to the ``flower structure'' of $G_{2} = G_{0}$, this follows from the fact that $G_{0}$ is a Frobenius group on $\Omega' = \Omega-\{0\}$. That is, $G_{0}$ is a transitive permutation group on $\Omega'$ in which the maximum number of fixed points of a nonidentity element is 1. (This is guaranteed by the sharp 3-transitivity of $G$.) Then by a theorem of Frobenius, the identity 1 of $G_{0}$, together with the set of elements with no fixed points, forms a normal subgroup $K$ called the \emph{Frobenius kernel} of the Frobenius group. In our case, $K$ is the subgroup $N$ in Sections~\ref{section:presentation} and \ref{section:pgl}, which is the unipotent radical of the Borel subgroup $G_{0}$ and is isomorphic to the additive group of the field $F$. Also $G_{0}-K$ is partitioned by the sets $G_{0,a}-\{1\}, a \in \Omega'$; these are the $|\Omega'| = q$ petals in the flower. The subgroups $G_{0,a}$ are the $q$ Cartan subgroups contained in $G_{0}$, and each is isomorphic to $F^{\times}$.

\subsection{Generalizations in $PGL(n,q)$}

Let $\tau = (G_i : 1\leq i\leq 4)$ be a family of subgroups of a finite group $G$. The Ingleton inequality~\eqref{equation:Ingletong} fails  iff
\[ |G_1G_2|<\frac{|G_{13}G_{23}||G_{14}G_{24}|}{|G_{34}|}. \]
In all constructions we will consider in this section, $G_i=G_{1i}G_{2i}$ for $i=3,4$ and $|G_3|=|G_4|$. Also $|G_1G_2|=|G_1:G_{12}||G_2|$. Hence in such constructions Ingleton is violated iff
\begin{equation}\label{equation:IngletonVioCond}
|G_1:G_{12}||G_2|<\frac{|G_3|^2}{|G_{34}|},
\end{equation}
and the Ingleton ratio~\eqref{equation:Ingletonr} becomes
\[ r(\tau)=\frac{|G_3|^2}{|G_1:G_{12}||G_2||G_{34}|}. \]

Now we explore three different approaches trying to extend the $PGL(2,q)$ family of violations $\rho$ to $PGL(n,q)$.

\subsubsection{Generalization 1}

Let $G=PGL(n,q)$ with $n\geq 3$. It is easy to see that $G$ is doubly transitive on the points of $PG(V)$ and transitive on triples of independent points. Let $P_i$, $2\leq i\leq 4$, be independent points in $V$, $\Delta_i=\{P_2,P_i\}$ for $i=3,4$, and $\Delta =\{P_2,P_3,P_4\}$. Set $G_2=N_G(P_2)$, $G_i=N_G(\Delta_i)$, $i=3,4$, and $G_1=N_G(\Delta)$. Let $\tau = (G_i : 1\leq i\leq 4)$.

Now $G_2$ is a parabolic and by Example~\ref{example:parabolic},
\begin{equation}\label{equation:3indepPtsG2}
|G_2|=q_nM_{n-1}.
\end{equation}
Next $D=P_2+P_3+P_4$ is a 3-dimensional subspace of $V$, so by Example~\ref{example:parabolic} again, $|N_G(D)|=q_nM_3M_{n-3}/(q-1)$. Further through calculation of the preimages in $GL(n,q)$ we have
\[ |N_G(D):G_1|=\frac{|GL(3,q)|}{6(q-1)^3} = \frac{q^3M_3}{6(q-1)^3}, \]
since $G_{1}$ acts as the symmetric group on $\Delta$ of order 3, and for each pair of points there are $q-1$ different choices of mappings. So
\begin{equation}\label{equation:3indepPtsG1}
|G_1|=\frac{|N_G(D)|\cdot 6(q-1)^3}{q^3M_3}=\frac{6q_nM_{n-3}(q-1)^2}{q^3}.
\end{equation}
As $G_1$ is transitive on $\Delta$ of order 3, $|G_1:G_{12}|=3$. Therefore
\begin{equation}\label{equation:3indepPtsLHS}
|G_1:G_{12}||G_2|=3|G_2|=3q_nM_{n-1}.
\end{equation}
Also for $i=3,4$, $G_{i}$ and $G_{1i}$ are both transitive on $\Delta_i$ of order 2, so $|G_{i}: G_{2i}| = |G_{1i}: G_{12i}| = 2$. Thus $|G_{1i}G_{2i}| = |G_{1i}: G_{12i}||G_{2i}| = |G_{i}|$ and  $G_i=G_{1i}G_{2i}$ for $i=3,4$. Since $G$ is doubly transitive on the points, $G_{3}$ is conjugate to $G_{4}$ and so $|G_3|=|G_4|$. Further $U=P_2+P_3$ is a 2-dimensional subspace of $V$, so by Example~\ref{example:parabolic}, $|N_G(U)|=q_nM_2M_{n-2}/(q-1)$. Also by calculating the preimages $|N_G(U):G_3|=|GL(2,q)|/(2(q-1)^2)=qM_2/(2(q-1)^2)$, so
\begin{equation}\label{equation:3indepPtsG3}
|G_3|=\frac{|N_G(U)|\cdot 2(q-1)^2}{qM_2}=\frac{2q_nM_{n-2}(q-1)}{q}.
\end{equation}
Finally $G_{34}=G_{\Delta}$ is the pointwise stabilizer of $\Delta$. Since $G_1$ is 3-transitive on $\Delta$, $|G_1:G_{34}| = 3! =6$. So by \eqref{equation:3indepPtsG1}:
\begin{equation}\label{equation:3indepPtsG34}
|G_{34}|=\frac{q_nM_{n-3}(q-1)^2}{q^3}.
\end{equation}
It follows from \eqref{equation:3indepPtsLHS}, \eqref{equation:3indepPtsG3}, and \eqref{equation:3indepPtsG34} that \eqref{equation:IngletonVioCond} is satisfied iff
\[ 3q_nM_{n-1}<\frac{4q_n^2M_{n-2}^2(q-1)^2\cdot q^3}{q^2\cdot q_nM_{n-3}(q-1)^2}=4q_nqM_{n-2}(q^{n-2}-1) \]
which holds iff $3(q^{n-1}-1)<4q(q^{n-2}-1)$ iff
\begin{equation}\label{equation:3indepPtsVioCond}
q^{n-1} - 4q + 3 > 0.
\end{equation}
This inequality holds when $n\geq 4$ or $n=3$ and $q\geq 4$.

Since $G$ is transitive on all triples of independent points, all 4-tuples in this generalization are conjugate to each other.

The Ingleton ratio is
\[ r(\tau) =\frac{4q_n^2M_{n-2}^2(q-1)^2\cdot q^3}{q^2\cdot 3q_nM_{n-1}\cdot q_nM_{n-3}(q-1)^2}=\frac{4q(q^{n-2}-1)}{3(q^{n-1}-1)}, \]
which approaches $4/3$ for large $q$ or $n$. Whereas in the original instance $\rho$, $r(\rho)=4(q-1)/(3q)$, which has the same asymptotics. But the scaling factors for both the Ingleton score and the violation index are usually larger than $PGL(2,q)$, so in general $\tau$ is less effective in violating Ingleton.

\subsubsection{Generalization 2}

As usual let $F=\f_{q}$ and $G=PGL(n,q)$, with $n\geq 2$. Let $P_i$, $2\leq i\leq 4$, be distinct but dependent points in $V$. Thus $P_{i}=Fx_{i}$, $i=2,3$, for two independent vectors $x_{2},x_{3} \in V$, and $P_{4}=Fx_{4}$, where $x_{4} = ex_{2}+x_{3}$ for some $e\in F$. Let $U$, $\Delta$, $\Delta_i$, $i=3,4$, and $G_i$, $1\leq i\leq 4$, be defined the same as in Generalization~1. Note that when $n=2$ this is our original construction $\rho$.

From Generalization~1, $|G_2|=q_nM_{n-1}$ and $|N_G(U)|=q_nM_2M_{n-2}/(q-1)$. Since $U$ is a 2-dimensional subspace of $V$, $PGL(U)$ is sharply 3-transitive on the points of $U$ by Lemma~\ref{lem:sh3trans}. Now as $\Delta$ is a set of three distinct points in $U$, its global stabilizer in $PGL(U)$ is isomorphic to $S_{3}$. Thus $G_{1}$ is 3-transitive on $\Delta$. Observe that each vector in $\{x_{i}:2\leq i\leq 4\}$ is a unique linear combination of the other two, with both coefficients nonzero. Then fixing a permutation of $\{P_i:2\leq i\leq 4\}$, there are only $q-1$ linear transformations in $GL(U)$ that respect this permutation. Hence $|N_G(U):G_1|=|GL(2,q)|/(6(q-1)) = qM_2/(6(q-1))$, and
\begin{equation}\label{equation:3depPtsG1}
|G_1|=\frac{|N_G(U)|\cdot 6(q-1)}{qM_2}=\frac{6q_nM_{n-2}}{q}.
\end{equation}
$G_1$ is transitive on $\Delta$, while for $i=3,4$, $G_{i}$ and $G_{1i}$ are both transitive on $\Delta_i$. $G$ is doubly transitive on the points of $PG(V)$. Thus from arguments in Generalization~1 we have $|G_1:G_{12}||G_2|=3q_nM_{n-1}$, $G_i=G_{1i}G_{2i}$ for $i=3,4$, and $|G_3|=|G_4|$. Also $|G_3|=2q_nM_{n-2}(q-1)/q$. Since $G_{34}=G_{\Delta}$ is of index 6 in $G_1$, by \eqref{equation:3depPtsG1}:
\[ |G_{34}|=\frac{q_nM_{n-2}}{q}. \]
Thus \eqref{equation:IngletonVioCond} is satisfied iff
\[ 3q_nM_{n-1}<\frac{4q_n^2M_{n-2}^2(q-1)^2\cdot q}{q^2\cdot q_nM_{n-2}}=\frac{4q_nM_{n-2}(q-1)^2}{q} \]
which holds iff $3q(q^{n-1}-1)<4(q-1)^2$ iff
\begin{equation}\label{equation:3depPtsVioCond}
3q\sum_{i=0}^{n-2}q^{i} - 4q + 4 < 0.
\end{equation}
When $n=2$, this inequality holds iff $q> 4$. When $n>2$, however, it always fails because $3q^{2} - q + 4 > 0$ for all $q$.

Therefore, the original instance $\rho$ is the only successful case in this construction, with Ingleton ratio $r(\rho)=4(q-1)/(3q)$.

\subsubsection{Generalization 3}

Again take $G=PGL(n,q)$ with $n\geq 3$. Let $U_2$ be a point of $V$, $U_i$, $i=3,4$, distinct 2-dimensional subspaces of $V$ with $U_3\cap U_4=U_2$, and $U_1=U_3+U_4$ the 3-dimensional subspace of $V$ generated by $U_3$ and $U_4$. Set $G_i=N_G(U_i)$ for $1\leq i\leq 4$, and $\lambda =(G_i : 1\leq i\leq 4)$. Then all the $G_i$ are parabolics with $|G_2|=q_nM_{n-1}$ from \eqref{equation:3indepPtsG2}, $|G_3|=|G_4|=q_nM_2M_{n-2}/(q-1)$, and $|G_1|=q_nM_3M_{n-3}/(q-1)$. As $G_1$ is transitive on the $(q^{3}-1)/(q-1) = q^2+q+1$ points in $U_1$, $|G_1:G_{12}|=q^2+q+1$, so
\[ |G_1:G_{12}||G_2|=(q^2+q+1)q_nM_{n-1}. \]
For $i=3,4$, $G_{i}$ and $G_{1i}$ are both transitive on the $(q^{2}-1)/(q-1) = q+1$ points in $U_i$, so $G_i=G_{1i}G_{2i}$ for $i=3,4$. Also $G_{34}$ is the subgroup of $G$ fixing $U_2$ and the points $U_3/U_2$ and $U_4/U_2$ of the quotient space $U_1/U_2$; in particular it is a subgroup of $G_{1}$. If we pick a basis $\X_{1}=\{x_{3},x_{2},x_{4}\}$ for $U_{1}$ such that $U_{2}=\langle x_{2} \rangle$ and $U_{i}=\langle x_{2}, x_{i} \rangle$ for $i=3,4$, then elements of $G_{34}$ correspond to the linear transformations in $GL(U_{1})$ whose matrices with respect to $\X_{1}$ take the form
\[ \begin{bmatrix}
a & 0 & 0 \\
x & b & y \\
0 & 0 & c
\end{bmatrix}, \]
where $a,b$ and $c$ are nonzero. So $|G_1:G_{34}|=|GL(3,q)|/(q^2(q-1)^3)=qM_3/(q-1)^3$, and
\[ |G_{34}|=\frac{|G_1|}{qM_3/(q-1)^3}=\frac{q_nM_3M_{n-3}\cdot (q-1)^3}{(q-1)\cdot qM_3}=\frac{q_nM_{n-3}(q-1)^2}{q}. \]
It follows that \eqref{equation:IngletonVioCond} is satisfied iff
\[ (q^2+q+1)q_nM_{n-1}<\frac{q_n^2M_2^2M_{n-2}^2\cdot q}{(q-1)^2\cdot q_nM_{n-3}(q-1)^2}=q_nq(q+1)^2(q^{n-2}-1)M_{n-2}, \]
which holds iff $(q^2+q+1)(q^{n-1}-1)<q(q+1)^2(q^{n-2}-1)$ iff
\[ q^n-q^3-q^2+1>0, \]
which holds iff $n\geq 4$.

The Ingleton ratio is
\[ r(\lambda) = \frac{q_n^2M_2^2M_{n-2}^2\cdot q}{(q-1)^2\cdot (q^2+q+1)q_nM_{n-1}\cdot q_nM_{n-3}(q-1)^2}=\frac{q(q+1)^2(q^{n-2}-1)}{(q^2+q+1)(q^{n-1}-1)}, \]
which approaches $1$ for large $q$ and $(q+1)^2/(q^2+q+1)$ (which is smaller than $4/3$) for large $n$. So this generalization seems less effective than the other two.

\subsection{Generalizations in General 2-transitive Groups}
\label{subsec:2-trans}

In the following we generalize the Ingleton violation $\rho$ in $PGL(2,q)$ to a more abstract construction, which includes Generalizations~1 and 2 as special cases.

Let $G$ be a doubly transitive group on a set $\Omega$ of order $l\geq 3$, let $\alpha$ and $\beta$ be distinct points in $\Omega$, and assume $\gamma\in\Omega -\{\alpha,\beta\}$ such that the global stabilizer $G(\Delta)$ of $\Delta =\{\alpha,\beta,\gamma\}$ acts as the symmetric group on $\Delta$ (which is clearly the case when $G$ is 3-transitive). Let $G_2=G_{\alpha}$, $G_3=G(\alpha,\beta)$, $G_4=G(\alpha,\gamma)$, and $G_1=G(\Delta)$. Set $\mu = (G_i : 1\leq i\leq 4)$.

Let $k=\big|G_{\alpha,\beta}\big|$, $d=|G_{\Delta}|$, $\Gamma$ the orbit of $\gamma$ under the action of $G_{\alpha,\beta}$, and $c=|\Gamma|$. Observe that $c=|G_{\alpha,\beta}:G_{\Delta}|=k/d$ and $c\leq l-2$ as $\Gamma\subseteq\Omega -\{\alpha,\beta\}$. Further $c=l-2$ iff $G$ is 3-transitive.

Since $G$ is 2-transitive on $\Omega$, $G_2$ is transitive on $\Omega -\{\alpha\}$ and so $|G_2:G_{\alpha,\beta}| = l-1$. Also $|G_1:G_{12}|=3$ as $G_1$ is transitive on $\Delta$, thus 
\[ |G_1:G_{12}||G_2|=3|G_2|=3(l-1)k. \]
Next $G_{3}$ is conjugate to $G_{4}$ by 2-transitivity of $G$ and for $i=3,4$, $G_{i}$ and $G_{1i}$ are both transitive on $\Delta_i$ of order 2, so $G_{1i}G_{2i}=G_i$ and $|G_i|=2k$ for $i=3,4$. Finally $G_{34}=G_{\Delta}$ is of order $d$. Thus
\[ |G_3|^2/|G_{34}|=4k^2/d=4kc, \]
so condition \eqref{equation:IngletonVioCond} is satisfied iff $3(l-1)k<4kc$ iff
\begin{equation}\label{equation:2transVioCond}
3(l-1) < 4c.
\end{equation}
Further the Ingleton ratio $r(\mu)=4c/(3(l-1))$.

If $G$ is 3-transitive then $c=l-2$, so $3(l-1)<4c=4(l-2)$ iff $l>5$. Further $r(\mu)=4(l-2)/(3(l-1))$.

Both Generalization~1 and 2 fit in this construction, with  $\rho$ being the only 3-transitive case. In Generalization~1, $l=(q^n-1)/(q-1)$ and by independence of points in $\Delta$,
\[ c=\frac{(q^n-1)-(q^2-1)}{q-1}=\frac{q^2(q^{n-2}-1)}{q-1}, \]
so by \eqref{equation:2transVioCond}, \eqref{equation:IngletonVioCond} is satisfied iff
\[ 3(\frac{q^n-1}{q-1}-1)<\frac{4q^2(q^{n-2}-1)}{q-1}, \]
which gives \eqref{equation:3indepPtsVioCond}. In Generalization~2, $l$ has the same value, but since $GL(U)$ is 3-transitive on the $(q^{2}-1)/(q-1) = q+1$ points of $U$, $c = q+1-2 = q-1$. Then by \eqref{equation:2transVioCond}, \eqref{equation:IngletonVioCond} is satisfied iff
\[ 3(\frac{q^n-1}{q-1}-1)<4(q-1), \]
which gives \eqref{equation:3depPtsVioCond}.

We see that the 3-transitive groups give rise to simple and effective Ingleton violation constructions. This category of groups include the alternating and symmetric groups, the groups $PGL(2,q)$ with $l=q+1$, the Mathieu groups, the affine groups of degree $2^e$ (which are the semidirect product of an $e$-dimensional vector space $E$ over $\mathbb{F}_2$ by $GL(E)$), and the subgroup of the affine group for $e=4$ where the complement is $A_7$ rather than $GL(4,2)\cong A_8$.

\section{Considerations for Constructing Group Network Codes}
\label{section:gnc}

We can use our Ingleton-violating groups to build group network codes. From Appendix~\ref{section:gnc-detail}, the resulting entropy vectors are characterizable by the subgroups used, thus they are capable of violating the Ingleton inequality. In contrast, the entropy vectors of linear network codes always respect Ingleton. Furthermore, let $G$  be any of $PGL(n,p)$, $PGL(n,q)$, $GL(n,p)$ or $GL(n,q)$. We will show in the following that linear network codes can be embedded in the group network codes constructed with direct products of copies of $G$. Apparently a direct product of any copies of an Ingleton-violating group still violates Ingleton, thus such classes of group network codes are strictly more powerful than linear network codes.

To construct a group network code, the choices of subgroups are not arbitrary: they should meet requirements (R1)--(R3). In particular, (R1) limits what subgroups can be associated with the sources: they need to satisfy
\begin{equation}\label{equation:indep}
\prod_{s\in\Sr}|G_s| = |G|^{|\Sr|-1}|G_\Sr|.
\end{equation}
When this is the case, we simply say the subgroups $\{G_s:s\in\Sr\}$ are independent in $G$. We will study the constructions of independent source subgroups in the context of $PGL(2,q)$ and $GL(2,q)$ (since they have simpler structures than the other higher-degree linear groups), and also provide a universal source subgroup construction for direct products of groups.

\subsection{Embeddings of Linear Network Codes}
\label{subsec:embedding}

As discussed in Appendix~\ref{subsec:lnc-inclusion}, linear network codes are a special type of group network codes. In particular, they are determined by the underlying additive group structure. The direct sum $V$ of source vector spaces can be called the \emph{ambient vector space} of a linear network code. Let $(V,+)$ denote the additive group of $V$. If we can find a finite group $G$ such that $(V,+)\leq G$, then the linear network code is said to be \emph{embedded} in the group network codes using $G$, since we can use subgroups of $G$ to construct an equivalent group network code.

Consider a linear network code with ambient vector space $V=\f_q^d$ for some $d$ and $q$, where $q=p^m$ for some prime $p$ and some integer $m$. Observing that $\f_q$ is an $m$-dimensional vector space over $\f_p$, we can establish the following facts:
\begin{enumerate}
	\item[i)] $\left(\f_p,+\right) \cong \z_p$,
	\item[ii)] $\left(\f_q,+\right) \cong \left(\f_p,+\right)^m \cong \z_p^m$,
	\item[iii)] $(V,+) \cong \left(\f_q,+\right)^d \cong \z_p^{md}$.
\end{enumerate}
Thus $(V,+)$ is embedded in the direct product of $m \cdot d$ copies of a group $G$, provided that $G$ contains an element of order $p$---by Cauchy's theorem, this condition is equivalent to $p$ divides $|G|$. It then follows that linear network codes over $\f_q$ are embedded in the group network codes using direct products of copies of $G^m$. In particular, let $G$ be any of the linear groups $PGL(2,p)$, $PGL(2,q)$, $GL(2,p)$ or $GL(2,q)$. We have the following embeddings in these groups, using properties of the matrix $A$ and the subgroup $N$:
\begin{enumerate}
	\item In $PGL(2,p)$, $\big|\A\big| = p$. So $(V,+) \cong \big\langle\A\big\rangle^{md} \leq PGL(2,p)^{md}$.
	\item In $GL(2,p)$, $|A| = p$. So $(V,+) \cong \langle A\rangle^{md} \leq GL(2,p)^{md}$.
	\item In $PGL(2,q)$, $N = \left\{\left.\As{\alpha}\, \right| \alpha\in\f_q\right\} \cong \z_p^m$. So $(V,+) \cong N^d \leq PGL(2,q)^d$.
	\item In $GL(2,q)$, $N = \left\{\left.A_\alpha \right| \alpha\in\f_q\right\} \cong \z_p^m$. So $(V,+) \cong N^d \leq GL(2,q)^d$.
\end{enumerate}
Therefore, we also have the corresponding network code embeddings. Furthermore, these  results for the degree-$2$ linear groups are easily extended to degree $n$, since the former are subgroups of the latter.

\subsection{Sources Independence Requirement Considerations}

If we want to utilize the Ingleton-violating groups $PGL(2,q)$ and $GL(2,q)$ to construct network codes, we need to find their independent subgroups. GAP searching shows that up to conjugation, $PGL(2,5)$ has 16 independent pairs of subgroups, 1 triple and no quadruple. For $GL(2,5)$, the numbers are 86, 14 and 0, respectively. It might be desirable to use some of the Ingleton-violating subgroups as sources, but we find no independent pairs in any violation instance in either $PGL(2,5)$ or $GL(2,5)$. Furthermore, we can prove the following negative results:

\begin{lemma}
Let $i,j \in \{1,2,3,4\}$ and $(i,j) \ne (3,4)$. For four random variables $X_1$, $X_2$, $X_3$ and $X_4$, if $X_i$ and $X_j$ are independent, then the Ingelton inequality \eqref{equation:Ingletonh} is satisfied.
\end{lemma}
\begin{proof}
By symmetry of \eqref{equation:Ingletonh}, we only need to prove the result for when $(i,j) = (1,2)$ or $(1,3)$. In the first case, $h_{12} = h_1 + h_2$, so
\begin{align*}
h_{12}+h_{13}+h_{14}+h_{23}+h_{24} &\geq  h_1 + h_2 + h_3 + h_{123} + h_4 + h_{124}\\
&\geq h_1+h_2+h_{34}+h_{123}+h_{124},
\end{align*}
where we used $h_{13}+h_{23} \geq h_3+h_{123}$ and $h_{14}+h_{24} \geq h_4+h_{124}$ by submodularity of entropy.
The second case is similar.
\end{proof}

\begin{corollary}
There is no independent triple or quadruple in a set of four subgroups that violates \eqref{equation:Ingletong}.
\end{corollary}

On another note, if we want to use the Ingleton-violating subgroups in the network, Proposition~\ref{proposition:source_intersection} in Appendix~\ref{section:gnc-detail} tells us that their intersection should contain the intersection of all the source subgroups. Since in $PGL(2,q)$ the intersection of the Ingleton-violating subgroups is trivial, we need to find trivially intersecting independent subgroups to serve as sources. In $PGL(2,5)$, there are 4 such pairs and no such triples. At least one of these pairs also extends to a general family:

\begin{proposition}
Let $U = \begin{bmatrix} 0 & -1 \\ t & 0 \end{bmatrix}\in GL(2,q)$, where $t$ is a primitive element in $\f_{q}$. Let $H$ be the image of $SL(2,q)$ in $PGL(2,q)$ under the natural homomorphism, which is isomorphic to $PSL(2,q)$. When $p\neq2$, $H$ and $\big\langle \U \big\rangle$ are independent in $PGL(2,q)$ with trivial intersection.
\end{proposition}
\begin{proof}
It is easy to see $\big|\U\big| = 2$, $\det U = t$. The determinant of any matrix representing an element in $H$ takes the form $t^{2k}\in \langle t^2\rangle$, for some $k$. But $t \notin \langle t^2\rangle$ as $q-1$ is even, so $H \bigcap\big\langle \U \big\rangle = 1$. Also $\big|\big\langle \U \big\rangle\big|\cdot|H| = {2\cdot|SL(2,q)|/2} = |SL(2,q)| = |PGL(2,q)|$, thus \eqref{equation:indep} holds.
\end{proof}

In $GL(2,q)$ there are more Ingleton-violating instances, which have
various intersections. So the requirement on the sources is not so strict
and we have a richer class of subgroups to work with. As in $PGL(2,q)$, there
exist trivially intersecting independent pairs, for example:

\begin{proposition}
In $GL(2,q)$, $SL(2,q)$ and $\langle B \rangle$ (or $\langle P \rangle$) are independent with trivial intersection.
\end{proposition}
\begin{proof}
Obviously $\det B^k = 1$ iff $B^k = I$, so $SL(2,q)$ and $\langle B \rangle$ have trivial intersection. Also $|B|\cdot|SL(2,q)| = {(q-1)}\cdot|GL(2,q)|/(q-1) = |GL(2,q)|$, thus \eqref{equation:indep} is satisfied. The proof for $\langle P \rangle$ is similar.
\end{proof}

In general it is not easy to find many independent subgroups in a group. If the group is a direct product of $n$ of its subgroups, however, it admits a natural construction of $n$ independent subgroups:

\begin{proposition}\label{proposition:drctprod}
If $G = G_1\times G_2\times\cdots\times G_n$, then ${1 \times G_2\times\cdots\times G_n}$, $G_1 \times 1\times\cdots\times G_n$, \ldots, and $G_1 \times G_2\times\cdots\times 1$ are $n$ trivially intersecting independent subgroups in $G$.
\end{proposition}
\begin{proof}
Trivial intersection is obvious, and it is easy to check that both sides of \eqref{equation:indep} are equal to $\prod^n_{i=1}|G_i|^{n-1}$.
\end{proof}


This construction is the generalization of the source construction for linear network codes, in which case the subgroup at source $s$ is the $W_{s}$ defined in Appendix~\ref{subsec:lnc-inclusion}. Also we see that using direct products we can obtain independent subgroups for an arbitrary number of sources, but the group order also grows.

If we further require the sources to be of the same alphabet size, then the independent subgroups must have the same order. In the above proposition, this can be simply achieved by choosing $G_i$ to be the same subgroup for each $i$. 
Additionally, for an arbitrary pair of independent subgroups, we have the following proposition.

\begin{proposition}
If $G_s$ and $G_r$ are independent in $G$, then $G_s\times G_r$ and $G_r\times G_s$ are independent in $G^2$ with the same order.
\end{proposition}
\begin{proof}
$G_s$ and $G_r$ satisfy $|G_s||G_r|=|G||G_s\bigcap G_r|$. Thus for the direct product construction, the $LHS$ and $RHS$ of \eqref{equation:indep} are $|G_s|^2|G_r|^2$ and $|G|^2|G_s\bigcap G_r|^2$ respectively, which are equal.
\end{proof}

\section{Conclusion}
\label{section:conclusion}

Using a refined search we find the smallest group to violate the Ingleton inequality to be the 120 element group $S_5$. Investigating the detailed structure of the subgroups allowed us to determine that this is an instance of the Ingleton-violating family of groups $PGL(2,q)$ for prime powers $q\geq 5$. As this family has a nice interpretation in the theory of group actions, we generalize the idea to obtain more Ingleton violations in $PGL(n,q)$ and $GL(n,q)$. We also examine the preimage group $GL(2,q)$ of $PGL(2,q)$ and discover more families of violating subgroups. Nevertheless, even in $PGL(2,q)$ and $GL(2,q)$ for $q>5$, there might still exist more violation instances that we have not explored, let alone other interesting groups. For example, subsequent to our work Boston and Nan\cite{Boston-Nan-NewVioIngelton} find many new violations in the class of permutation groups. Presumably there are infinite families of Ingleton violating groups, so the list of such families to date is by no means comprehensive and is far from complete.

The $PGL$ and $GL$ groups violate the Ingleton inequality and, since they contain linear network codes inside them, can provide network codes more powerful than linear ones. Developing group network codes requires designing the source subgroups that satisfy independence (R1) and the edge subgroups that satisfy (R2) and (R3). The coding process requires two fundamental operations: (i) determining the intersection of all cosets from each incoming edge, and (ii) finding the appropriate coset for the outgoing overgroup of the intersected subgroups. Therefore constructing network codes from $PGL(n,q)$ and $GL(n,q)$ will require a thorough understanding of the structure of their subgroups and the corresponding coset operation. Investigating this issue may be a fruitful direction for future work.


%

\appendices
%

\section{Group Network Codes: Details}
\label{section:gnc-detail}

\subsection{Code Construction}

To establish the encoding and decoding process, we need an auxiliary lemma.
\begin{lemma}\label{lemma:coset_mapping}
Let $K_{1}, K_{2}$ be two subgroups of $G$ with $K_{1}\leq K_{2}$. Then the coset mapping
\begin{equation}\label{equation:coset_mapping}
\begin{array}{c}
\pi:G/K_{1}\to G/K_{2} \\
\phantom{\pi:}xK_{1} \mapsto xK_{2}
\end{array}
\end{equation}
is a well defined onto function, where $xK_{1}$ is mapped to the unique coset in $G/K_{2}$ that contains it. Furthermore, if $\Lambda_{1}$ is a uniform random variable on $G/K_{1}$, then $ \pi(\Lambda_{1})$ is uniform on $G/K_{2}$.
\end{lemma}
\begin{proof}
$\pi$ is well defined since $xK_{2}=x'K_{2}$ whenever $xK_{1}=x'K_{1}$. Note that $K_{2}$ is partitioned by the $m$ distinct cosets $\{y_{i}K_{1}: 1\leq i\leq m\}$, where $m = |K_{2}/K_{1}|$ and $y_{i}\in K_{2}$ for $i=1,2,\ldots,m$. Therefore, each $xK_{2}\in G/K_{2}$ is also partitioned by the $m$ cosets $\{(xy_{i})K_{1}: 1\leq i\leq m\}$, which are precisely the $m$ preimages of $xK_{2}$ under $\pi$. Thus $\pi(\Lambda_{1})$ is uniform on $G/K_{2}$.
\end{proof}

For any collection $\alpha$ of subgroups of $G$, the intersection mapping~\eqref{equation:intersection_mapping} is a bijection. Consider the collection of all source subgroups. Let $\X_{\Sr} = {\{(xG_s: s\in\Sr)\ |\ x\in G\}}\subseteq\prod_{s\in\Sr}\Y_{s}$, then we have the bijective intersection mapping $\Theta_{\Sr}: \X_{\Sr} \to G/G_\Sr$. But with (R1), $\left|\prod_{s\in\Sr}\Y_{s}\right| = |G/G_{\Sr}| = |\X_{\Sr}|$ and so
\[ \X_{\Sr}=\prod_{s\in\Sr}\Y_{s}. \]
This means that any coset tuple $(x_{s}G_s: s\in\Sr)$ in $\prod_{s\in\Sr}\Y_{s}$ can be represented in the form $(xG_s: s\in\Sr)$ for a common $x\in G$, and the intersection of $\{x_{s}G_s: s\in\Sr\}$ is equal to $xG_{\Sr}$. Therefore, we can rewrite the bijection $\Theta_{\Sr}$ as
\[ \Theta_{\Sr}: \prod_{s\in\Sr}\Y_{s} \to G/G_\Sr, \]
which maps a tuple to the intersection of all its cosets.

Moreover, let $t$ be an edge or a sink node, define $\X_{\In(t)} = \{(xG_f:f\in\In(t))\ |\ x\in G\}$ and $G_{\In(t)} = \bigcap_{f\in\In(t)}G_f$. Then the intersection mapping
\[ \Theta_{\In(t)}: \X_{\In(t)} \to G/G_{\In(t)} \]
is a bijection. With (R2) and (R3), we can also define coset mappings for edges and source/sink pairs as follows. For each edge $e$, since $G_{\In(e)}\leq G_{e}$ by (R2), define the coset mapping $\pi_{e}$ as \eqref{equation:coset_mapping} with $K_{1}=G_{\In(e)}$ and $K_{2}=G_{e}$. Similarly for each source $s$ with $u\in\D(s)$, since $G_{\In(u)}\leq G_{s}$ by (R3), define $\pi_{u,s}$ with $K_{1}=G_{\In(u)}$ and $K_{2}=G_{s}$.

Now we can define the encoding and decoding functions. At each edge $e$, let the encoding function be $\phi_{e} = \pi_{e}\circ\Theta_{\In(e)}$. For each source $s$ with $u\in\D(s)$, let the decoding function be $\phi_{u,s} = \pi_{u,s}\circ\Theta_{\In(u)}$. In other words, at an edge or a sink node $t$, the encoding/decoding function takes an input coset tuple $(Y_{f}:f\in\In(t))$ and first forms the intersection of them, which is a coset of $G_{\In(t)}$, then maps this coset to the unique coset of $G_{e}$ (or $G_{s}$, whichever is appropriate) that contains it. Such network operations define a proper network code, since by the proposition below the decoding functions always yield correct source symbols at each sink node.

\begin{proposition}\label{proposition:consistent_egde_symbols}
Assume (R1) holds, and let the encoding and decoding functions be defined as above. Then for some common $x \in G$, $\forall s\in\Sr$, $Y_{s} = xG_s$ and $\forall e\in\E$, $Y_{e} = xG_{e}$. Also for each source $s$ with $u\in\D(s)$, $Y_{s}$ is recovered by the decoding function $\phi_{u,s}$.
\end{proposition}
\begin{proof}
Let the source symbols $(Y_s: s\in\Sr)$ be an arbitrary tuple from $\prod_{s\in\Sr}\Y_{s}$. Since (R1) is true, as discussed above, for all $s\in\Sr$, $Y_{s} = xG_s$ with a common $x\in G$. As $\G$ is directed and acyclic, we can define the ``depth'' of each node $v$ as the length of the longest path from a source node to $v$, and define the depth of an edge to be the depth of its tail node. Note that  $e$ is always ``deeper'' than $f$ if $f\in\In(e)$.  Also if $Y_f = xG_{f}$ for all $f\in\In(e)$, then $Y_e = \phi_{e}(Y_f:f\in\In(e)) = xG_{e}$. So by induction on the depths of the edges, $Y_{e} = xG_{e}$ for all $e\in\E$.

Furthermore, for each $s\in\Sr$ with $u\in\D(s)$, since $Y_f = xG_{f}$ for all $f\in\In(u)$, $\phi_{u,s}{(Y_f:f\in\In(u))} = xG_{s} = Y_{s}$. Thus the source symbol $Y_{s}$ is successfully recovered at $u$.
\end{proof}

\begin{remark}
Note that the encoding/decoding function for an edge or a sink node $t$ is only defined on $\X_{\In(t)}$, but not on the entire Cartesian product $\prod_{f\in\In(t)}\Y_{f}$. This is because for an arbitrary tuple in $\prod_{f\in\In(t)}\Y_{f}$, it is possible that the intersection of all cosets is the empty set, which is not a coset of $G_{\In(t)}$. However, with (R1) this is not a problem, as Proposition~\ref{proposition:consistent_egde_symbols} guarantees that $(Y_{f}:f\in\In(t))$ is always a tuple in $\X_{\In(t)}$.
\end{remark}

\begin{remark}
From the proof above, even without (R1) these encoding and decoding functions still constitute a valid network code, if the sources cooperate in such a way that the transmit tuples are always from $\X_{\Sr}$. But in this case the source random variables are dependent.
\end{remark}

\subsection{The Entropy Vector}

Here we analyze the global mappings of this group network code, and show that the entropy vector is characterizable by the group $G$ and its subgroups $\{G_{t} : t\in\Sr\cup\E\}$ when the sources are independent and uniform. First we give another auxiliary lemma. 

\begin{lemma}\label{lemma:entropy_vector_pi}
Let $K\leq G$ and let $G_{i}$, $i=1,\ldots,n$, be subgroups of $G$ containing $K$. For each $i$ let $\pi_{i}$ be the coset mapping defined as \eqref{equation:coset_mapping} with $K_{1}=K$ and $K_{2}=G_{i}$. Let $\Lambda_{K}$ be a uniform random variable on $G/K$, and define $X_{i} = \pi_{i}(\Lambda_{K})$ for each $i$. Then the entropy vector of $\{X_1,X_2,\dots,X_n\}$ is exactly the group characterizable vector induced by $G$ and $\{G_1,G_2,\dots,G_n\}$.
\end{lemma}
\begin{proof}
For each nonempty subset $\alpha \subseteq \n$, since $K\leq G_{\alpha}$, we can define the coset mapping $\pi_{\alpha}$ with $K$ and $G_{\alpha}$. As in Section~\ref{subsec:group}, the alphabet of $X_\alpha$ is still $\X_{\alpha} = \{(xG_i:i\in\alpha)\ |\ x\in G\}$, and the intersection mapping $\Theta_{\alpha}$ is a bijection. Also $\Theta_{\alpha}(X_\alpha) = \pi_{\alpha}(\Lambda_{K})$, which is uniform on $G/G_{\alpha}$ by Lemma~\ref{lemma:coset_mapping}. So the joint entropy $H(X_\alpha) = H(\Theta_{\alpha}(X_\alpha) ) = \log\frac{|G|}{|G_{\alpha}|}$ and the lemma follows.
\end{proof}

For each $s\in\Sr$ define the coset mapping $\pi'_{s}$ as \eqref{equation:coset_mapping} with $K_{1}=G_{\Sr}$ and $K_{2}=G_{s}$. For every edge $e$ we can similarly define a new coset mapping $\pi'_{e}$ with $K_{1}=G_{\Sr}$ and $K_{2}=G_{e}$, since according to the following proposition, $G_\Sr \leq G_e$.

\begin{proposition}\label{proposition:source_intersection}
If (R2) is satisfied, then $\forall e\in\E$, $G_\Sr \leq G_e$.
\end{proposition}
\begin{proof}
The proposition is trivially true if $e$ is emitted from a source node. Also if $G_\Sr \leq G_f$ for all $f\in\In(e)$, then by (R2) we have $G_\Sr \leq G_e$. Similar to Proposition~\ref{proposition:consistent_egde_symbols}, by induction on the depths of the edges the proof follows.
\end{proof}

\begin{proposition}\label{proposition:global_mapping}
$\forall e\in\E$, the global mapping at $e$ for the above group network code is $\varphi_{e} = \pi'_{e}\circ\Theta_{\Sr}$. In other words, $\varphi_{e}$ first forms the intersection of all the source cosets to obtain a coset of $G_\Sr$, and then maps this coset to the unique coset of $G_{e}$ containing it.
\end{proposition}
\begin{proof}
Assume the source symbols $(Y_s: s\in\Sr)$ are transmitted and let $\Lambda_{\Sr} = \Theta_{\Sr}(Y_s: s\in\Sr)$. Then $\Lambda_{\Sr} = xG_{\Sr}$ for some $x\in G$, and $Y_{s} = xG_s = \pi'_{s}(\Lambda_{\Sr})$ for all $s\in\Sr$. By Proposition~\ref{proposition:consistent_egde_symbols}, $Y_{e} = xG_e = \pi'_{e}(\Lambda_{\Sr})$, so
$\varphi_{e} = \pi'_{e}\circ\Theta_{\Sr}$.
\end{proof}

Let the source random variables $\{Y_s: s\in\Sr\}$ be independent and uniformly distributed, so the joint distribution is uniform on $\prod_{s\in\Sr}\Y_{s}$. Let $\Lambda_{\Sr} = \Theta_{\Sr}(Y_s: s\in\Sr)$, then $\Lambda_{\Sr}$ is uniform on $G/G_{\Sr}$ as $\Theta_{\Sr}$ is bijective. From Proposition~\ref{proposition:global_mapping}, $\forall t\in\Sr\cup\E$, $Y_{t} = \pi'_{t}(\Lambda_{\Sr})$, and so by Lemma~\ref{lemma:entropy_vector_pi}, the entropy vector for $\{Y_{t}: t\in\Sr\cup\E\}$ is characterizable by the group $G$ and its subgroups $\{G_{t} : t\in\Sr\cup\E\}$.

\subsection{Inclusion of Linear Network Codes}\label{subsec:lnc-inclusion}

In this section we carry over the group theory notations in Section~\ref{section:notation} to vector spaces, but with additive notation. For example, the left coset is now written as $v+W$ for a vector $v$ and a subspace $W$. Further, we use $\oplus$ to denote the direct sum of vector spaces. In the following we show that for each linear network code, there exists an equivalent group network code, with essentially the same network operations and hence the same encoding/decoding results.

Consider a linear network code $\cd$ over a finite field $F$. For each $t\in\Sr\cup\E$, the alphabet $\Y_{t}$ is a finite dimensional vector space over $F$. Let $v$ denote the concatenation of all the source vectors $(Y_s: s\in\Sr)$, then $v$ is a vector in $V \triangleq \oplus_{s \in \Sr\;} U_s$, where $U_{s} \triangleq \Y_{s}$. Then for each edge $e$, the global mapping $\varphi_{e}$ is a linear transformation from $V$ to $\Y_{e}$, whose range is denoted by $U_{e}$. Also for each source $s$, let $\varphi_{s}:V\to U_{s}$ be the linear projection that maps $v\in V$ to its $s$-th section. Thus $\forall t\in\Sr\cup\E$, we can write $Y_{t} = \varphi_{t}(v)$. Let $W_{t}$ be the null space of $\varphi_{t}$, then by the First Isomorphism Theorem,
\[ \psi_{t} : v+W_{t} \mapsto \varphi_{t}(v) \]
is a vector space isomorphism between the quotient space $V/W_{t}$ and $U_{t}$.

Let $t$ be an edge or a sink node. If $Y_{f}=0$ for all $f\in\In(t)$, then $Y_{t} = 0$ as the encoding/decoding functions are linear. Thus $\bigcap_{f\in\In(t)}W_{f} \leq W_{t}$. Further, for each source $s$
\[ W_{s} = \{ v\in V\ |\ s\text{-th section of v is 0}\} \cong \oplus_{r \in \Sr\setminus\{s\}\;} U_r, \]
so $\bigcap_{s\in\Sr}W_{s} = 0$. Since $V/W_{s}\cong U_{s}$, we have $\prod_{s\in\Sr}|V/W_{s}| = |V|$. Let $G = V$, $G_t = W_{t}$ for all $t\in\Sr\cup\E$. As $V$ is a finite dimensional vector space over a finite field, $G$ is a finite group. It is straightforward to check that the requirements (R1)--(R3) are all satisfied, so we can define a group network code $\cd'$ with these groups.

This network code is equivalent to $\cd$, since $\{\psi_{t}: t\in\Sr\cup\E\}$ provides a set of bijections between their codewords at each source/edge, and these bijections respect the encoding/decoding operations. In particular, assume in $\cd$ the source vectors yield some $v\in V$, and so $Y_{t} = \varphi_{t}(v)$ is transmitted at each source/edge $t$. Then with $\psi_{t}$ the corresponding symbol for $\cd'$ is $v+W_{t}$, which is consistent with the encoding/decoding result of $\cd'$ at each edge/sink node by Proposition~\ref{proposition:consistent_egde_symbols}.

\begin{figure}[!t]
\centering
\input{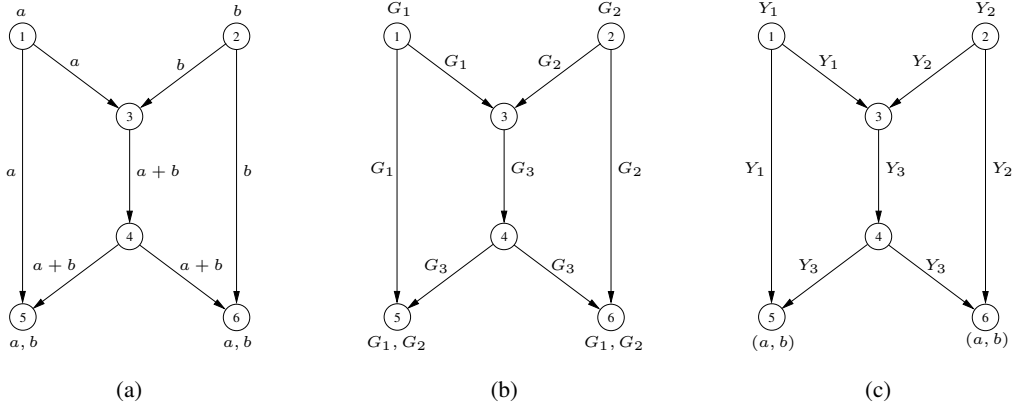}
\caption{Two network codes on the the butterfly network. (a) A linear network code; (b) the subgroup assignment for the corresponding group network code; (c) the transmitted symbols in the group network code. In (b), $G=\{(a,b): a,b\in\f_q\}$, $G_1=\{(0,x): x\in\f_q\}$, $G_2=\{(y,0): y\in\f_q\}$, and $G_3=\{(z,-z): z\in\f_q\}$. In (c), $Y_1 = \{(a,x): x\in\f_q\}$, $Y_2=\{(y,b): y\in\f_q\}$, and $Y_3=\{(a+z,b-z): z\in\f_q\}$.}
\label{figure:butterfly}
\end{figure}
For example, Fig.~\ref{figure:butterfly} demonstrates a linear network code over $\f_{q}$ for the well-known butterfly network (Fig.~\ref{figure:butterfly}-(a)), and the corresponding group network code (Fig.~\ref{figure:butterfly}-(b),(c)). Here for the linear network code, we have $V=\f_q^2$, $U_{1} = U_{2} = U_{e_{34}} = \f_{q}$, while $W_1=\{(0,x): x\in\f_q\}$, $W_2=\{(y,0): y\in\f_q\}$, and $W_{e_{34}}=\{(z,-z): z\in\f_q\}$. If we set $G=V$, $G_1=W_1$, $G_2=W_2$, and $G_3=W_{e_{34}}$, then the resulting group network code is equivalent to the original linear one.

\section{Proofs and Calculations in Section~\ref{section:gl}}

\subsection{Structures of $M,K,K',J,J'$}
\label{subsec:app_subgrp}

When the characteristic $p$ of $\f_q$ equals $2$, $K=K'$ and $J=J'$. So for the analysis of $K'$ and $J'$ we only consider the case $p\neq2$.

Observe that $|A_\alpha|=p$ for each $\alpha\in\f_q^\times$, and
\[ |C|=3,\quad |B_1|=2,\quad |B|=|B'|=|P|=|P'|=q-1. \]
As $(CB_1)^2 = I$, we have $M \cong D_6 \cong S_3$. It is easy to check that $\forall\alpha\in\f_q$,
\[ A_\alpha^B = A_{t^{-1}\alpha},\quad A_\alpha^{B'} = A_{-t^{-1}\alpha},\quad
A_\alpha^P = A_{t\alpha},\quad A_\alpha^{P'} = A_{-t\alpha}. \]
Therefore, $N$ is a normal subgroup of all $K,K',J,J'$ and
\[ K = N\cdot\langle B\rangle,\quad K' = N\cdot\langle B'\rangle,\quad
J = N\cdot\langle P\rangle,\quad J' = N\cdot\langle P'\rangle. \]
Also $N$ trivially intersects each of $\langle B\rangle, \langle B'\rangle, \langle P\rangle$ and $\langle P'\rangle$, thus
\[ K\cong N\rtimes\langle B\rangle,\quad K'\cong N\rtimes\langle B'\rangle,\quad
J\cong N\rtimes\langle P\rangle,\quad J'\cong N\rtimes\langle P'\rangle, \]
all of which are semidirect products $\z_p^m\rtimes\z_{q-1}$. We claim that $K\cong J$ and $K'\cong J'$. Moreover, in the case $p\neq2$, all the four groups are isomorphic if and only if $\fq$ is even.

To see this, first consider the bijections $\sigma: K\to J$ and $\sigma': K'\to J'$, where $\forall\alpha\in\f_q$, $\forall k\in\K_q$,
\[ \sigma\left(A_\alpha B^k\right) = A_\alpha P^{-k},\quad \sigma'\left(A_\alpha (B')^k\right) = A_\alpha (P')^{-k}. \]
Observe that $\forall\alpha,\beta\in\f_q$, $\forall k,l\in\K_q$,
\[ \sigma\left(A_\alpha B^k\cdot A_\beta B^l\right) = \sigma\left(A_{\alpha+t^k\beta}B^{k+l}\right) =
A_{\alpha+t^k\beta}P^{-k-l} = A_\alpha P^{-k}\cdot A_\beta P^{-l} =
\sigma\left(A_\alpha B^k\right)\cdot \sigma\left(A_\beta B^l\right), \]
so $\sigma$ is indeed an isomorphism. Similarly $\sigma'$ is also an isomorphism.

Next observe that in the case $p\neq2$, when $\fq$ is even, $\frac{q-1}{4}$ is an integer and so
\[ \left(\mfq\right)^2 = \left(\fq+1\right)^2 = \frac{(q-1)^2}{4} + (q-1) + 1 \equiv 1 \pmod{q-1}. \]
Thus $\left((B')^\mfq\right)^\mfq = B'$ and $\big\langle (B')^\mfq\big\rangle = \langle B'\rangle$. In addition, since $\f_q^\times$ is cyclic of an even order $q-1$, we have $-1=t^\fq$, and thus $(-t)^\mfq = \left(t^\mfq\right)^\mfq = t$. Consider $\tau: K\to K'$, where
\[ \tau\left(A_\alpha B^k\right) = A_\alpha (B')^{\mfq k},\quad \forall\alpha\in\f_q,\ \forall k\in\K_q. \]
Apparently $\tau$ is a bijection. Also we can show that it is a homomorphism by calculating $\tau\left(A_\alpha B^k\cdot A_\beta B^l\right)$ with the following fact:
\[ A_\alpha(B')^{\mfq k}\cdot A_\beta(B')^{\mfq l} = A_{\alpha+(-t)^{\mfq k}\beta}(B')^{\mfq(k+l)}
= A_{\alpha+t^k\beta}(B')^{\mfq(k+l)}. \] 
Thus when $\fq$ is even, $K\cong K'$ and the four groups are all isomorphic.

When $\fq$ is odd, however, $\tau$ is not a bijection anymore, because this time $B'\notin\big\langle (B')^\mfq\big\rangle$ and $\tau(K)\neq K'$. Furthermore, we can prove that in this case $K$ and $K'$ are not isomorphic, by showing that $K$ and $J$ have generalized flower structures whenever $q>2$, whereas if $p\neq2$, $K'$ and $J'$ only have flower structures when $\fq$ is even. Since $K\cong J$ and $K'\cong J'$, it is enough to only show the analysis of $K$ and $K'$. Pick $\alpha\in\f_q^\times$ and assume $k,l\in\K_q$. Similar to the $G_2$ in Section~\ref{subsec:pgl(2,q)}, we have the relation
\[ (B^k)^{A_\alpha} = B^l \iff k = l = 0, \]
thus $K$ has a generalized flower structure whenever $q>2$. On the other hand, for $K'$ we have
\[ (B'^k)^{A_\alpha} = B'^l \iff \begin{bmatrix} (-1)^k & 0 \\ t^k\alpha & t^k \end{bmatrix} = 
\begin{bmatrix} (-1)^l & 0 \\ (-1)^l\alpha & t^l \end{bmatrix}, \]
which requires $k = l$ and $t^l = (-1)^l$. Thus for $p\neq2$, $l$ can only be 0 or $\fq$. If $\fq$ is even, we have $(-1)^\fq = 1$ and so $k = l = 0$, then $K'$ also has a generalized flower structure (as expected since here $K\cong K'$). If $\fq$ is odd, however, this is not true: in this case $(-1)^\fq=-1$, so $k = l = 0$ or $\fq$ in the above relation. Thus $\forall\alpha\in\f_q^\times$, $\langle B'\rangle \bigcap \langle B'\rangle^{A_\alpha} = \langle -I \rangle \cong \z_2$. When $q=3$, $B' = -I$ and $K' = \langle A \rangle \times \langle -I \rangle \cong \z_3\times\z_2 \cong \z_6$; when $q>3$, $\langle B' \rangle$ and $\langle B' \rangle^{A_\alpha}$ are distinct groups but have nontrivial intersection. Therefore, in neither case does $K'$ have a generalized flower structure.

\subsection{Intersections in Instances 8 and 9}
\label{subsec:app_89int}

Let $p\neq2$. Observe that $K'$ and $J'$ are both subgroups of the $G_2$ in Instance~1, so all the intersections in both instances are subgroups of their respective counterparts in Instance~1. In Instance~8, since ${G_{12}\leq\langle tI, B_1 \rangle}$ and the (1,1)-entry for every matrix in $G_2=K'$ is always $\pm 1$, we have $G_{12}\leq\langle -I, B_1 \rangle$. This further limits the (2,2)-entry to be $\pm 1$ for each matrix in $G_{12}$. As the (2,2)-entry in $K'$ takes the form $t^k$ for some $k$, this $k$ can only be $0$ or $\fq$. By examining the parity of $\fq$, we have
\[ G_{12} = \left\{\begin{array}{cl}\langle B_1 \rangle \cong \z_2 & \text{if }\fq\text{ is even}\\
    \langle -I \rangle \cong \z_2 & \text{otherwise}\end{array}\right.,\quad
G_{123} = G_{124} = \left\{\begin{array}{cl}1 & \text{if }\fq\text{ is even}\\
    \langle -I \rangle \cong \z_2 & \text{otherwise}\end{array}\right.. \]
Similarly we can calculate $G_{12},G_{123}$ and $G_{124}$ for Instance~9.

In both instances, $G_{24}$ is simply the subgroup of all diagonal matrices in $G_2$, and $G_{23}\leq T$. As matrices in $K'$ and $J'$  can be respectively written as
\[ (-1)^k \begin{bmatrix} 1 & 0 \\ \alpha' & (-t)^k \end{bmatrix} =
(-1)^k \begin{bmatrix} 1 & 0 \\ \alpha' & (t^\mfq)^k \end{bmatrix}\quad
\text{and}\quad t^k\begin{bmatrix} 1 & 0 \\ \alpha'' & (-t^{-1})^k \end{bmatrix} =
t^k\begin{bmatrix} 1 & 0 \\ \alpha'' & (t^{\frac{q-3}{2}})^k \end{bmatrix} \]
for some $\alpha',\alpha''\in\f_q$ and $k\in \K_q$, we see that $G_{23} = \langle -B_3^\mfq \rangle$ and
$\langle tB_3^{\frac{q-3}{2}} \rangle$ respectively, where
\[ (-B_3^\mfq)^k = \begin{bmatrix} (-1)^k & 0 \\ t^k-(-1)^k & t^k \end{bmatrix},\quad
(tB_3^{\frac{q-3}{2}})^k = \begin{bmatrix} t^k & 0 \\ (-1)^k - t^k & (-1)^k \end{bmatrix}. \]
Thus $G_{23} \cong \z_{q-1}$ in both cases.

\subsection{The case $p=3$ for Instance 15}
\label{subsec:app_15p3}

In Instance~15, $G_1 = M = \langle C, B_1\rangle$ and $G_2 = (J')^E$. We can show that $G_1G_2 = G_2G_1$ when $p=3$, thus Condition~\ref{condition:g1g2} is satisfied. Observe that $G_2 = \left\{X_{\alpha,j}\,|\,\alpha\in\f_q, j\in\K_q \right\}$, where
\[ X_{\alpha,j} \triangleq
\begin{bmatrix} (-1)^j - \alpha & \alpha \\ (-1)^j - t^j - \alpha & t^j + \alpha \end{bmatrix}. \]
When $p=3$, we have $2=-1$. With this relation, it is easy to check that $C = X_{1,0} \in G_2$, and for each $\alpha$ and $j$
\[ X_{\alpha,j}^{B_1} =
\begin{bmatrix} (-1)^j + \alpha & -\alpha \\ (-1)^j - t^j + \alpha & t^j - \alpha \end{bmatrix}
= X_{-\alpha,j}\in G_2. \]
Thus $G_1$ normalizes $G_2$. In particular, $\forall X \in G_2$ and $\forall Y \in G_1$, we have $X^Y\in G_2$ and $X^{Y^{-1}}\in G_2$, which imply $XY\in G_1G_2$ and $YX\in G_2G_1$ respectively. Therefore $G_1G_2 = G_2G_1$.

\subsection{Intersections in Instances 12--15}
\label{subsec:app_1215int}

Most intersections are easily obtained by comparing the formulae of the matrices in the subgroups involved. For the intersection of $M$ with any of $J^E,(J')^E,J^Q$ or $(J')^Q$, we can utilize the properties below to facilitate calculation. Let $\vec{c}_i(X)$ denote the $i$-th column of a matrix $X$, we have
\[ \vec{c}_1(X)+\vec{c}_2(X)=\begin{bmatrix} 1 \\ 1 \end{bmatrix},\ \forall X\in J^E;\qquad
\vec{c}_1(X)+\vec{c}_2(X)=\pm\begin{bmatrix} 1 \\ 1 \end{bmatrix},\ \forall X\in (J')^E; \]
\[ \vec{c}_1(X)-2\vec{c}_2(X)=\begin{bmatrix} 1 \\ -2 \end{bmatrix},\ \forall X\in J^Q;\qquad
\vec{c}_1(X)-2\vec{c}_2(X)=\pm\begin{bmatrix} 1 \\ -2 \end{bmatrix},\ \forall X\in (J')^Q. \]
Thus, we need only seek elements of $M$ which share these properties.

We also want to mention the calculation of $G_{34}$ for Instances 13 and 15 when $p>3$. In Instance 13, finding $G_{34}$ is equivalent to solving the following set of equations:
\[ \left\{\begin{array}{c}(-1)^j-\alpha = (-1)^i+2\beta \\
\alpha = \beta \\
(-1)^j-t^j-\alpha = 2\left(t^i-2\beta-(-1)^i\right) \\
t^j+\alpha = t^i-2\beta
\end{array}\right. \iff
\left\{\begin{array}{c}\alpha = \beta \\
3\beta  = (-1)^j-(-1)^i \\
t^i = (-1)^j \\
t^j = (-1)^i
\end{array}\right.. \]
From the last two equations, we can see that $i$ and $j$ can only be 0 or $\fq$. If $\fq$ is even, then $(-1)^\fq=1$, so $i$ and $j$ must both be 0, which yields that $G_{34} = 1$. If $\fq$ is odd, then $i=0$ implies that $j=0$, and $i=\fq$ implies that $j=\fq$. In both cases $\alpha=\beta=0$, therefore $G_{34} = \langle -I \rangle$. For $G_{34}$ in Instance 15, we have similar equations and the same discussion also applies.

\section*{Acknowledgment}

The authors would like to thank Michael~Aschbacher and Amin~Shokrollahi for very helpful discussions on the conditions and the generalizations of the violations,
and on expanding the group structures. They would also like to
thank Ryan~Kinser for reminding them of the case $PGL(2,q)$.

\ifCLASSOPTIONcaptionsoff
  \newpage
\fi



\bibliographystyle{IEEEtran}
\bibliography{IEEEabrv,IngletonVioIT}
\end{document}